\documentclass[11pt]{article}
\def\withcolors{0}
\def\withnotes{0}
\def\withindex{1}
 \makeatletter{}\usepackage[T1]{fontenc}
\usepackage[utf8]{inputenc}

\usepackage{lmodern}
\usepackage{xspace}                                    
\usepackage[normalem]{ulem}

\usepackage{amsfonts,amsmath,amssymb, amsthm, mathtools}
\usepackage{thm-restate}
\usepackage{dsfont} 
\usepackage{algorithmicx,algpseudocode,algorithm}

\usepackage[usenames,dvipsnames,table]{xcolor}

\usepackage{csquotes}

\usepackage{relsize}

\usepackage{multirow}
\usepackage{chngpage} 
\usepackage{mfirstuc}

\usepackage{tikz}
\usetikzlibrary{arrows}
\usetikzlibrary{calc,decorations.pathmorphing,patterns}

\ifnum\withindex=1
  \usepackage{makeidx}
  \usepackage{ifthen}
  \newcommand\indexed[2][]{\ifthenelse{\equal{#1}{}}{#2\index{#2}}{#2\index{#1}}}
  \makeindex \fi

\usepackage[backref,colorlinks,citecolor=blue,bookmarks=true,linktocpage]{hyperref}
\usepackage{aliascnt}
\usepackage[numbered]{bookmark}

\usepackage{fullpage}

\usepackage{titling}

\usepackage[shortlabels]{enumitem}
  \setitemize{noitemsep,topsep=3pt,parsep=2pt,partopsep=2pt}   \setenumerate{itemsep=1pt,topsep=2pt,parsep=2pt,partopsep=2pt}
  \setdescription{itemsep=1pt}
  
\ifnum\withnotes=1
  \usepackage[colorinlistoftodos,textsize=scriptsize]{todonotes}
\fi

\usepackage{verbatim}

\usepackage{mleftright}  
\makeatletter{}\makeatletter
\@ifundefined{theorem}{    \theoremstyle{plain}   	\newtheorem{theorem}{Theorem}[section]
  	\newaliascnt{coro}{theorem}
  	  \newtheorem{corollary}[coro]{Corollary}
  	\aliascntresetthe{coro}
  	\newaliascnt{lem}{theorem}
  		\newtheorem{lemma}[lem]{Lemma}
  	\aliascntresetthe{lem}
  	\newaliascnt{clm}{theorem}
  		\newtheorem{claim}[clm]{Claim}
	\aliascntresetthe{clm}
	\newaliascnt{fact}{theorem}
 	 	\newtheorem{fact}[theorem]{Fact}
	\aliascntresetthe{fact}
  	
  \newaliascnt{prop}{theorem}
  		\newtheorem{proposition}[prop]{Proposition}
	\aliascntresetthe{prop}
	\newaliascnt{conj}{theorem}
  		\newtheorem{conjecture}[conj]{Conjecture}
	\aliascntresetthe{conj}
  \theoremstyle{remark}   	\newtheorem{remark}[theorem]{Remark}
  	\newtheorem{question}[theorem]{Question}

  \theoremstyle{definition}   	\newaliascnt{defn}{theorem}
 		 \newtheorem{definition}[defn]{Definition}
 	 \aliascntresetthe{defn}
}{}
\makeatother

\newenvironment{proofof}[1]{\begin{proof}[Proof of {#1}]}{\end{proof}}

\providecommand{\email}[1]{\href{mailto:#1}{\nolinkurl{#1}\xspace}}

\ifnum\withcolors=1
  \newcommand{\new}[1]{{\color{red} {#1}}}             \else
  \newcommand{\new}[1]{{{#1}}}

\fi

\ifnum\withnotes=1
     \newcommand{\todonote}[2][]{\todo[size=\scriptsize,color=red!40,#1]{#2}}

\else
  
  \newcommand{\todonote}[2][]{\ignore{#2}}

\fi
\newcommand{\ignore}[1]{\leavevmode\unskip}  
\newcommand{\eps}{\ensuremath{\varepsilon}\xspace}
\newcommand{\Algo}{\ensuremath{\mathcal{A}}\xspace} \newcommand{\Tester}{\ensuremath{\mathcal{T}}\xspace}  \newcommand{\property}{\ensuremath{\mathcal{P}}\xspace}  \newcommand{\eqdef}{\stackrel{\rm def}{=}}

\newcommand{\accept}{\textsf{ACCEPT}\xspace}
\newcommand{\fail}{\textsf{FAIL}\xspace}
\newcommand{\reject}{\textsf{REJECT}\xspace}

\newcommand{\half}{\frac{1}{2}}
\newcommand{\domain}{\ensuremath{\Omega}\xspace}  \newcommand{\yes}{{\sf{}yes}\xspace}
\newcommand{\no}{{\sf{}no}\xspace}

\newcommand{\littleO}[1]{{o\mleft( #1 \mright)}}
\newcommand{\bigO}[1]{{O\mleft( #1 \mright)}}
\newcommand{\bigTheta}[1]{{\Theta\mleft( #1 \mright)}}
\newcommand{\bigOmega}[1]{{\Omega\mleft( #1 \mright)}}
\newcommand{\tildeO}[1]{\tilde{O}\mleft( #1 \mright)}

\newcommand{\tildeOmega}[1]{\operatorname{\tilde{\Omega}}\mleft( #1 \mright)}
\providecommand{\poly}{\operatorname*{poly}}

\newcommand{\setOfSuchThat}[2]{ \left\{\; #1 \;\colon\; #2\; \right\} } 			\newcommand{\indicSet}[1]{\mathds{1}_{#1}}                                              \newcommand{\indic}[1]{\indicSet{\left\{#1\right\}}}                                             
\newcommand{\dtv}{\operatorname{d_{\rm TV}}}
\newcommand{\totalvardist}[2]{{\dtv\!\left({#1, #2}\right)}}

\newcommand\restr[2]{{  \left.\kern-\nulldelimiterspace   #1   \vphantom{\big|}   \right|_{#2}   }}

\newcommand{\proba}{\Pr}

\newcommand{\probaDistrOf}[2]{\proba_{#1}\left[\, #2\, \right]}

\newcommand{\shortexpect}{\mathbb{E}}

\newcommand{\uniform}{\ensuremath{\mathcal{U}}}
\newcommand{\uniformOn}[1]{\ensuremath{\uniform\!\left( #1 \right) }}

\newcommand{\norm}[1]{\lVert#1{\rVert}}
\newcommand{\normone}[1]{{\norm{#1}}_1}

\newcommand{\abs}[1]{\left\lvert #1 \right\rvert}

\newcommand{\vect}[1]{\mathbf{#1}} 			

\newcommand{\flr}[1]{\left\lfloor #1 \right\rfloor}

\newcommand{\ICOND}{{\sf INTCOND}\xspace}
\newcommand{\EVAL}{{\sf EVAL}\xspace}
\newcommand{\CDFEVAL}{{\sf CEVAL}\xspace}

\newcommand{\SAMP}{{\sf SAMP}\xspace}
\newcommand{\COND}{{\sf COND}\xspace}
\newcommand{\PCOND}{{\sf PAIRCOND}\xspace}
\newcommand{\ORACLE}{{\sf ORACLE}\xspace}

\newcommand{\pdfsamp}{dual\xspace}
\newcommand{\cdfsamp}{cumulative dual\xspace}
\newcommand{\Pdfsamp}{Dual\xspace}\newcommand{\Cdfsamp}{Cumulative Dual\xspace}
\newcommand{\lp}[1][1]{\ell_{#1}}

\newcommand{\D}{\ensuremath{D}}
\newcommand{\distrD}{\ensuremath{\mathcal{D}}}
\newcommand{\birge}[2][\D]{\Phi_{#2}(#1)}
\newcommand{\iid}{i.i.d.\xspace}

\makeatletter

\newcommand{\Rom}[1]{\expandafter\@slowromancap\romannumeral #1@}

\makeatother

\makeatletter
\newcommand\ackname{Acknowledgements}
\if@titlepage
  \newenvironment{acknowledgements}{      \titlepage
      \null\vfil
      \@beginparpenalty\@lowpenalty
      \begin{center}        \bfseries \ackname
        \@endparpenalty\@M
      \end{center}}     {\par\vfil\null\endtitlepage}
\else
  
\fi
\makeatother

\usepackage{graphicx} \newcommand{\mirror}[1]{\reflectbox{#1}}
\usepackage{soul}
\definecolor{lightblue}{rgb}{.90,.88,1}
\sethlcolor{lightblue}

\newcommand{\Dflat}{\birge[\D]}
\newcommand{\Dred}{{\D}^\mathrm{red}}

\def\maintitle{Big Data on the Rise?}
\def\subtitle{Testing monotonicity of distributions}
\def\authorname{Cl\'ement L. Canonne}

\title{\textsc{\maintitle}\\\textsmaller[2]{(\subtitle)} }
\author{\authorname\thanks{Columbia University. Email: \email{ccanonne@cs.columbia.edu}. Research supported by NSF CCF-1115703 and NSF CCF-1319788.}}

\makeatletter
  \hypersetup{
    pdftitle={\maintitle{} (\subtitle)},
    pdfauthor={\authorname}, 
    pdfsubject={Probability Distribution Testing},
    pdfkeywords={property testing} {distribution testing} {sublinear algorithms} {monotonicity}
  }
\makeatother

\begin{document}

\maketitle

\begin{abstract}
The field of property testing of probability distributions, or distribution testing, aims to provide fast and (most likely) correct answers to questions pertaining to specific aspects of very large datasets. In this work, we consider a property of particular interest, \emph{monotonicity of distributions}. We focus on the complexity of monotonicity testing across different models of access to the distributions \cite{CFGM:13,CRS:12,CR:14,RS:09}; and obtain results in these new settings that differ significantly from the known bounds in the standard sampling model~\cite{BKR:04}.
\end{abstract}
\pagenumbering{gobble}\setcounter{page}{0}
\clearpage

\setcounter{tocdepth}{2}  \tableofcontents
\ifnum\withnotes=1
  \listoftodos
\fi
\clearpage
\pagenumbering{arabic}
\section{Introduction}
\makeatletter{}Before even the advent of data, information, records and insane amounts thereof to treat and analyze, probability distributions have been everywhere, and understanding their properties has been a fundamental problem in Statistics.\footnote{As well as~--~crucially~--~in crab population analysis~\cite{Pearson:94}.} Whether it be about the chances of winning a (possibly rigged) game in a casino, or about predicting the outcome of the next election; or for social studies or experiments, or even for the detection of suspicious activity in networks, hypothesis testing and density estimation have had a role to play. And among these distributions, \emph{monotone} ones have often been of paramount importance: is the probability of getting a cancer decreasing with the distance from, say, one's microwave? Are aging voters more likely to vote for a specific party? Is the success rate in national exams correlated with the amount of money spent by the parents in tutoring?

All these examples, however disparate they may seem, share one unifying aspect: \emph{data} may be viewed as the probability distributions it defines and originates from; and understanding the properties of this data calls for testing these distributions. In particular, our focus here will be on testing whether the data~--~its underlying distribution~--~happens to be \emph{monotone},\footnotemark{} or on the contrary far from being so.

Since the seminal work of Batu, Kumar, and Rubinfeld~\cite{BKR:04}, this fundamental property has been well-understood in the usual model of access to the data, which only assumes independent samples. However, a recent trend in distribution testing has been concerned with introducing and studying new models which provide additional flexibility in observing the data. In these new settings, our understanding of what is possible and what remains difficult is still in its infancy; and this is in particular true for monotonicity, for which very little is known. This work intends to mitigate this state of affairs.
\footnotetext{Recall that a distribution $\D$ on $\{1,\dots,n\}$ is said to be \emph{monotone} (non-increasing) if $\D(1)\geq \dots \geq \D(n)$, i.e. if its probability mass function is non-increasing. We hereafter denote by $\mathcal{M}$ the class of monotone distributions.}

We hereafter assume the reader's familiarity with the broad field of property testing, and the more specific setting of distribution testing. For detailed surveys of the former, she or he is referred to, for instance,~\cite{Survey:Fischer,Survey:Ron:08,Survey:Ron:10,Survey:Goldreich:10}; an overview of the latter can be found e.g.~in~\cite{Rubinfeld:12:Taming}, or~\cite{Canonne:15:survey}. Details of the models we consider (besides the usual sampling oracle setting, denoted by \SAMP) are described in~\cite{CFGM:13,CRS:12,CRS:14} (for the conditional sampling oracle \COND, and its variants \ICOND and \PCOND restricted respectively to interval and pairwise queries);~\cite{BKR:04,GMV:06,CR:14} for the \Pdfsamp and \Cdfsamp models; and~\cite{RS:09} for the evaluation-only oracle, \EVAL. The reader confused by the myriad of notations featured in the previous sentence may find the relevant definitions in~\autoref{sec:preliminaries} and~\autoref{app:def:models} (as well as in the aforementioned papers).

\paragraph{Caveat.} It is worth mentioning that we do not consider here monotonicity in its full general setting: we shall only focus on distributions defined on the line. In particular, the (many) works concerned with distributions over high-dimensional posets are out of the scope of this paper.

\subsection{Results}

In this paper, we provide both upper and lower bounds for the problem of testing monotonicity, across various types of access to the unknown distribution. A summary of results, including the best currently known bounds on monotonicity testing of distributions, can be found in~\autoref{table:summary:results} below. As noted in~\autoref{sec:samp:tester}, many of the lower bounds are implied by the corresponding lower bound on testing uniformity.
\newcommand{\newresult}[1]{\hl{#1}}
\newcommand{\nonadapt}{${}^\ast$\xspace}
  \begin{table}[H]\centering\def\arraystretch{1.75}
    \begin{tabular}{@{}|l|c|c|@{}}
    \hline
  {\sc Model}          & {\sc Upper bound}          & {\sc Lower bound} \\ \hline
  {\SAMP} & $\tildeO{\frac{\sqrt{n}}{\eps^6}}$ & $\bigOmega{\frac{\sqrt{n}}{\eps^2}}$ \\\hline
  {\COND} & \newresult{ $\tildeO{\frac{1}{\eps^{22}}}$, $\tildeO{ \frac{\log^2 n}{\eps^{3}}+\frac{\log^4 n}{\eps^{2}} }$ } & $\bigOmega{\frac{1}{\eps^2}}$ \\\hline
  {\ICOND} & \newresult{ $\tildeO{\frac{\log^5 n}{\eps^{4}} }$ }  & $\bigOmega{ \sqrt{\frac{\log n}{\log\log n}} }$ \\\hline
  {\EVAL} & \newresult{ $\bigO{\max\mleft( \frac{\log n}{\eps}, \frac{1}{\eps^2}\mright)}$}\nonadapt & \newresult{ $\bigOmega{\frac{\log n}{\eps}}$}\nonadapt,\newresult{ $\bigOmega{\frac{\log n}{\log\log n}}$} \\\hline
  {\Cdfsamp} & \newresult{ $\tildeO{\frac{1}{\eps^4}}$ } & $\bigOmega{\frac{1}{\eps}}$ \\\hline
    \end{tabular}
    \caption[Comparison between \SAMP, \COND, \ICOND, \EVAL and \Cdfsamp.]{\label{table:summary:results} Summary of results for monotonicity testing. The \newresult{highlighted ones} are new; bounds with an asterisk\nonadapt hold for non-adaptive testers.}
\end{table}

\subsection{Techniques}

Two main ideas are followed in obtaining our upper bounds: the first one, illustrated in~\autoref{sec:samp:tester} and~\autoref{sec:polyneps:cond:tester}, is the approach of Batu et al.~\cite{BKR:04}, which reduces monotonicity testing to uniformity testing on polylogarithmically many intervals. This relies on a structural result for monotone distributions which asserts that they admit a succinct partition in intervals, such that on each interval the distribution is either close to uniform (in $\lp[2]$ distance), or puts very little weight.

The second approach, on which~\autoref{sec:polyeps:cond:tester},~\autoref{sec:polylogneps:eval:ub} and~\autoref{sec:polyeps:extended:tester} are based, also leverages a structural result, due this time to Birg\'e~\cite{Birge:87}. As before, this theorem states that each monotone distribution admits a succinct ``flat approximation,'' but in this case the partition \emph{does not depend on the distribution itself} (see~\autoref{sec:preliminaries} for a more rigorous exposition). From there, the high-level idea is to perform two different checks: first, that the distribution $\D$ is close to its ``flattening''  $\bar{\D}$; and then that this flattening itself is close to monotone~--~where \new{to be efficient} the latter exploits the fact that the effective support of $\bar{\D}$ is very small, as there are only polylogarithmically many intervals in the partition. If both tests succeed, then it must be the case that $\D$ is close to monotone.

\subsection{Organization}

The upper bounds for the conditional models,~\autoref{theo:tester:monotonicity:cond} and~\autoref{theo:tester:monotonicity:intcond}, are given in~\autoref{sec:cond}. \autoref{sec:eval} contains details of the results in the evaluation query model: the upper bound of~\autoref{theo:tester:monotonicity:eval}, and the non-adaptive and adaptive lower bounds of respectively~\autoref{theo:tester:monotonicity:eval:lb} and~\autoref{theo:tester:monotonicity:eval:lb:adaptive}. Finally, in~\autoref{sec:polyeps:extended:tester} we prove~\autoref{theo:tester:monotonicity:extension:cdf}, our upper bound for the \Cdfsamp access model.
We also note that, in the course of obtaining one of our upper bounds, we derive a result previously (to the best of our knowledge) absent from the literature: namely, that learning monotone distributions in the \EVAL-only model can be accomplished using $\bigO{(\log n)/\eps}$ queries (\autoref{lemma:eval:learn:monotone}).

Finally, we show in~\autoref{app:tolerant:extended} that some of our techniques extend to \emph{tolerant} testing, and describe in two of the models tolerant testers for monotonicity whose query complexity is only logarithmic in $n$ (\autoref{theo:tolerant:tester:monotonicity} and~\autoref{theo:tolerant:tester:monotonicity:cdf}). 
 
\section{Preliminaries}\label{sec:preliminaries}
\makeatletter{}All throughout this paper, we denote by $[n]$ the set $\{1,\dots,n\}$, and by $\log$ the logarithm in base $2$. A \emph{probability distribution} over a (finite) domain\footnote{For the focus of this work, all distributions will be supported on a finite domain; thus, we do not consider the fully general definitions from measure theory.} $\domain$ is a non-negative function $\D\colon\domain\to[0,1]$ such that $\sum_{x\in\domain} \D(x) = 1$. We denote by $\uniformOn{\domain}$ the uniform distribution on \domain. Given a distribution $\D$ over $\domain$ and a set $S\subseteq\domain$, we write $\D(S)$ for the total probability weight $\sum_{x\in S} \D(x)$ assigned to $S$ by $\D$. Finally, for $S \subseteq \domain$ such that $\D(S)>0$, we denote by $\D_S$ the conditional distribution of $\D$ restricted to $S$, that is $\D_S(x) = \frac{\D(x)}{\D(S)}$ for $x \in S$ and $\D_S(x)=0$ otherwise.\medskip

As is usual in property testing of distributions, in this work the distance between two distributions $\D_1, \D_2$ on $\domain$ will be the \emph{total variation distance}:
\begin{equation}\label{def:distance:tv}
\totalvardist{\D_1}{\D_2 } \eqdef \frac{1}{2} \normone{\D_1 - \D_2} = \frac{1}{2} \sum_{x \in \domain}\abs{\D_1(i)-\D_2(i)} = \max_{S \subseteq \domain} (\D_1(S)-\D_2(S))
\end{equation}
which takes value in $[0,1]$.

\paragraph{Models and access to the distributions.} We now describe (informally) the settings we shall work in, which define the \emph{type of access} the testing algorithms are granted to the input distribution. (For a formal definition of these models, the reader is referred to~\autoref{app:def:models}.) In the first and most common setting (\SAMP), the testers access the unknown distribution by getting independent and identically distributed samples from it.

 A natural extension, \COND, allows the algorithm to provide a query set $S\subseteq[n]$, and get a sample from the conditional distribution induced by $\D$ on $S$: that is, the distribution $\D_S$ on $S$ defined by $\D_S(i) = \D(i)/\D(S)$. By restricting the type of allowed query sets to the class of intervals $\{a,\dots,b\}\subseteq[n]$, one gets a weaker version of this model, $\ICOND$ (for ``interval-cond'').
 
 Of a different flavor, providing (only) \emph{evaluation} queries to the probability mass function (pmf) (resp. to the cumulative distribution function (cdf)) of the distribution  an \EVAL (resp. \CDFEVAL) oracle access. When the algorithm is provided with both \SAMP and \EVAL (resp. \SAMP and \CDFEVAL) oracles to the distribution, we say it has \emph{\Pdfsamp (resp.~\Cdfsamp{}) access} to it.

\paragraph{On the domain and parameters.} Unless specified otherwise, $\domain$ will hereafter by default be the $n$-element set $[n]$. When stating the results, the accuracy parameter $\eps\in[0,1]$ is to be understood as taking small values, either a tiny constant or a quantity arbitrarily close to $0$; however, the actual parameter of interest will always be $n$, viewed as ``going to infinity.'' Hence any dependence on $n$, no matter how mild, shall be considered as more expensive than any function of $\eps$ only.

\paragraph{On monotone distributions.} We now state here a few crucial facts about monotone distributions, namely that they admit a succinct approximation, itself monotone, close in total variation distance:
\begin{definition}[Oblivious decomposition]\label{def:birge:obl:decomp}
  Given a parameter $\eps>0$, the corresponding \emph{oblivious decomposition of $[n]$} is the partition $\mathcal{I}_\eps=(I_1,\dots,I_\ell)$, where $\ell=\bigTheta{\frac{\ln( \eps n + 1)}{\eps}}=\bigTheta{\frac{\log n }{\eps}}$ and\footnote{We will often ignore the floors in the definition of the oblivious partition, to avoid more cumbersome analyses and the technicalities that would otherwise arise. However, note that this does not affect the correctness of the proofs: after the first $\tildeO{\frac{1}{\eps}}$ intervals (which will be, as per the above definition, of constant size), we do have indeed $\abs{I_{k+1}}\in[1+\frac{\eps}{2},1+2\eps]\abs{I_{k}}$. This multiplicative property, in turn, is the key aspect we shall rely on.} $\abs{I_{k}}=\flr{(1+\eps)^k}$, $1\leq k \leq \ell$. \end{definition}

\noindent For a distribution $\D$ and parameter \eps, define $\birge[\D]{\eps}$ to be the \emph{flattened distribution} with relation to the oblivious decomposition $\mathcal{I}_\eps$:
\begin{equation}\label{eq:birge:def:dflat}
  \forall k\in [\ell], \forall i\in I_k,\quad \birge[\D]{\eps}(i) = \frac{\D(I_k)}{\abs{I_k}}\;.
\end{equation}
Note that while $\birge[\D]{\eps}$ (obviously) depends on $\D$, \emph{the partition $\mathcal{I}_\eps$ itself does not}; in particular, it can be computed prior to getting any sample or information about $\D$.

\begin{theorem}[\cite{Birge:87}]\label{theorem:birge:obl:decomp}
 If $\D$ is monotone non-increasing, then $\totalvardist{\D}{\birge[\D]{\eps}} \leq \eps$.
\end{theorem}

\begin{remark} The first use of this result in this discrete learning setting is due to Daskalakis et al.~\cite{DDS:12}. For a proof for discrete distributions (whereas the original paper by Birg\'e is intended for continuous ones), the reader is referred to \cite{DDSV:13} (Section~3.1, Theorem~5).
\end{remark}

\begin{corollary}[Robustness]\label{coro:birge:decomposition:robust}
  Suppose $\D$ is \eps-close to monotone non-increasing. Then $\totalvardist{\D}{\birge[\D]{\alpha}} \leq 2\eps+\alpha$; furthermore,  $\birge[\D]{\alpha}$ is also \eps-close to monotone non-increasing.
\end{corollary}
\begin{proof}
Let $P$ be a monotone non-increasing distribution such that $\totalvardist{\D}{P} \leq \eps$. By the triangle inequality,
\[
    \totalvardist{\D}{\birge[\D]{\alpha}} \leq \totalvardist{\D}{P} + \totalvardist{P}{\birge[P]{\alpha}} + \totalvardist{\birge[P]{\alpha}}{\birge[\D]{\alpha}}
    \leq \eps + \alpha + \totalvardist{\birge[P]{\alpha}}{\birge[\D]{\alpha}}
\]
where the last inequality uses the assumption on $P$ and \autoref{theorem:birge:obl:decomp} applied to it. It only remains to bound the last term: by definition,
\begin{align*}
  2\totalvardist{\birge[P]{\alpha}}{\birge[\D]{\alpha}} &= \sum_{i=1}^n \abs{\birge[\D]{\alpha}(i) - \birge[P]{\alpha}(i)} 
  = \sum_{k=1}^\ell\sum_{i\in I_k} \abs{\birge[\D]{\alpha}(i) - \birge[P]{\alpha}(i)} \\
  &= \sum_{k=1}^\ell\sum_{i\in I_k} \abs{ \frac{D(I_k) - P(I_k)}{\abs{I_k}} } = \sum_{k=1}^\ell \abs{ D(I_k) - P(I_k) } \\
  &= \sum_{k=1}^\ell \abs{ \sum_{i\in I_k}\left(  D(i) - P(i)  \right)} 
  \leq \sum_{k=1}^\ell \sum_{i\in I_k} \abs{ D(i) - P(i)} = 2\totalvardist{P}{\D} \\
  &\leq 2\eps
\end{align*}
(showing in particular the second part of the claim, as $\birge[P]{\alpha}$ is monotone) and thus
\[
    \totalvardist{\D}{\birge[\D]{\alpha}} \leq 2\eps+\alpha
\]
as claimed.
\end{proof}
One can interpret this corollary as saying that the Birg\'e decomposition provides a tradeoff between becoming \emph{simpler} (and at least as close to monotone) while not staying too far from the original distribution.

\noindent Incidentally, the last step of the proof above implies the following easy fact:
\begin{fact}\label{fact:birge:shrinking}
  For all $\alpha\in(0,1]$,
  \begin{equation}
  \totalvardist{\birge[P]{\alpha}}{\birge[\D]{\alpha}} \leq \totalvardist{P}{\D}
  \end{equation}
  and in particular, for any property $\mathcal{P}$ preserved by the Birg\'e transformation (such as monotonicity)
  \begin{equation}
  \totalvardist{\birge[\D]{\alpha}}{\mathcal{P}} \leq \totalvardist{\D}{\mathcal{P}}.
  \end{equation}
\end{fact}

\paragraph{Other tools.}
Finally, we will use as subroutines the following results of Canonne, Ron, and Servedio. The first one, restated below, provides a way to ``compare'' the probability weight of disjoint subsets of elements in the \COND model:
\begin{lemma}[{\cite[Lemma 2]{CRS:12}}]\label{lem:cond:compare}
Given as input two disjoint subsets of points $X,Y\subseteq\domain$ together with parameters
$\eta \in (0,1]$, $K \geq 1$, and $\delta\in (0,1/2]$,
as well as \COND query access to a distribution $\D$ on $\domain$,
there exists a procedure {\sc Compare} that either outputs a value $\rho > 0$ or
outputs {\sf High} or {\sf Low}, and satisfies the following:
\begin{enumerate}[(i)]
  \item\label{compare:mid}
  If $\D(X)/K \leq \D(Y) \leq K \cdot \D(X)$ then with probability
  at least $1-\delta$
  the procedure outputs a value $\rho \in [1-\eta,1+\eta]\D(Y)/\D(X)$;

  \item\label{compare:high}
  If $\D(Y) > K \cdot \D(X)$ then with probability at least
  $1-\delta$ the procedure outputs
  either {\sf High} or a value $\rho \in [1-\eta,1+\eta]\D(Y)/\D(X)$;

  \item\label{compare:low}
  If $\D(Y) < \D(X)/K$ then with probability at least $1-\delta$ the procedure
  outputs either {\sf Low} or a value $\rho \in [1-\eta,1+\eta]\D(Y)/\D(X).$
\end{enumerate}
The procedure performs $\bigO{\frac{K\log(1/\delta)}{\eta^2} }$
\COND queries on the set $X\cup Y$.
\end{lemma}
The second allows one to estimate the distance between the uniform distribution and an unknown distribution $\D$, given access to a conditional oracle to the latter:
\begin{theorem}[{\cite[Theorem 14]{CRS:12}}]\label{lem:cond:estim}
Given as input $\eps \in (0,1]$ and $\delta\in (0,1]$, as well as \PCOND query access to a distribution $\D$ on $\domain$, there exists an algorithm that outputs a value $\hat{d}$ and has the following guarantee. The algorithm performs $\tildeO{1/\eps^{20} \log(1/\delta)}$ queries and, with probability at least $1-\delta$, the value it outputs satisfies $\abs{\hat{d} - \totalvardist{\D}{\uniform}}\leq \eps$.
\end{theorem}
 
\section{Previous work: Standard model}\label{sec:samp:tester}
In this section, we describe the currently known results for monotonicity testing in the standard (sampling) oracle model. These bounds on the sample complexity, tight up to logarithmic factors, are due to Batu et al.~\cite{BKR:04};\footnote{\cite{BKR:04} originally claims an $\tildeO{{\sqrt{n}}/{\eps^4}}$ sample complexity, but their argument seems to only result in an $\tildeO{{\sqrt{n}}/{\eps^6}}$ bound. Subsequent work building on their techniques~\cite{CDGR:15} obtains the $\eps^4$ dependence.} while not directly applicable to the other access models we will consider, we note that some of the techniques they use will be of interest to us in \autoref{sec:polyneps:cond:tester}.
\begin{theorem}[{\cite[Theorem 10]{BKR:04}}]
There exists an $\bigO{\frac{\sqrt{n}}{\eps^6}\poly\!\log n}$-query tester for monotonicity in the \SAMP model.
\end{theorem}
\begin{proof}[Proof (sketch)]
Their algorithm works by taking this many samples from $\D$, and then using them to recursively split the domain $[n]$ in half, as long as the conditional distribution on the current interval is not close enough to uniform (or not enough samples fall into it). If the binary tree created during this recursive process exceeds $\bigO{\log^2 n/\eps}$ nodes, the tester rejects. Batu et al. then show that this succeeds with high probability, the leaves of the recursion yielding a partition of $[n]$ in $\ell_{\max}=\bigO{\log^2 n/\eps}$ intervals $I_1,\dots,I_{\ell_{\max}}$, such that either
\begin{enumerate}[\sf(a)]
  \item the conditional distribution $\D_{I_j}$ is $\bigO{\eps}$-close to uniform on this interval; or
  \item $I_j$ is ``light,'' i.e. has weight at most $\bigO{\eps/\ell_{\max}}$ under $\D$.
\end{enumerate}
(the first item relying on a lemma from \cite{BFRSW:00} relating distance to uniformity and collision count\footnote{We observe that the dependence on $\eps$ could be brought down to $\eps^4$, by using instead machinery from \cite[Theorem 11]{DKN:15} to perform this step.}). This implies this partition defines an $\ell_{\max}$-flat distribution $\bar{\D}$ which is $\eps/2$-close to $\D$, and can be easily learnt from another batch of samples; once this is done, it only remains to test (e.g., via linear programming, which can be done efficiently) whether this $\bar{\D}$ is itself $\eps/2$-close to monotone, and accept if and only this is the case.
\end{proof}
\begin{theorem}[{\cite[Theorem 11]{BKR:04}}]
Any tester for monotonicity in the \SAMP model must perform $\bigOmega{\frac{\sqrt{n}}{\eps^2}}$ queries.
\end{theorem}
\begin{proof}[Proof (sketch)]
To prove this lower bound, they reduce the problem of uniformity testing to monotonicity testing: from a distribution $\D$ over $[n]$ (where $n$ is for the sake of simplicity assumed to be even), one can run a monotonicity tester (with parameter $\eps^\prime\eqdef\eps/3$) on both $\D$ and $\mirror{\D}$, where the latter is defined as $\mirror{\D}(i)\eqdef \D(n-i)$, $i\in[n]$;  and accept if and only if both tests pass. If $\D$ is uniform, clearly $\D=\mirror{\D}$ is monotone; conversely, one can show that if both $\D$ and its ``mirrored version'' $\mirror{\D}$ pass the test (are $\eps^\prime$-close to monotone non-increasing), then it must be the case that $\D$ is $\eps$--close to uniform. The lower bound then follows from the $\bigOmega{\frac{\sqrt{n}}{\eps^2}}$ lower bound of \cite{Paninski:08} for testing uniformity.\footnote{\cite{BKR:04} actually only shows a $\bigOmega{\sqrt{n}}$ lower bound, as they invoke in the last step the (previously best known) lower bound of \cite{GRexp:00} for uniformity testing; however, their argument straightforwardly extends to the result of Paninski.}
\end{proof}
We note that the argument above extends to all models: that is, any lower bound for testing uniformity directly implies a corresponding lower bound for monotonicity in the same access model (giving the bounds in \autoref{table:summary:results}).

\paragraph{Open question.}
At a (very) high-level, the above results can be interpreted as ``relating monotonicity to uniformity.'' That is, the upper bound is essentially established by proving that monotonicity reduces to testing uniformity on polylogarithmically many intervals, while the lower bound follows from showing that it reduces from testing uniformity on a constant number of them. Thus, an interesting question is whether, \emph{qualitatively}, the former or the latter is tight in terms of $n$. Are uniformity and monotonicity strictly as hard, or is there an intrinsic gap, even if only polylogarithmic, between the two?
\begin{question}\label{question:monotonicity:uniformity}
Can monotonicity be tested in the \SAMP model with $\bigO{\sqrt{n}}$ samples, or are $\bigOmega{\sqrt{n}\log^c n}$ needed for some absolute constant $c > 0$?
\end{question}
\section{With conditional samples}\label{sec:cond}
\makeatletter{}In this section, we focus on testing monotonicity with a stronger type of access to the underlying distribution, that is given the ability to ask conditional queries. More precisely, we prove the following theorem:
\begin{theorem}\label{theo:tester:monotonicity:cond}
There exists an $\tildeO{\frac{1}{\eps^{22}}}$-query tester for monotonicity in the \COND model.
\end{theorem}
\noindent Furthermore, assuming only a (restricted) type of conditional queries are allowed, one can still get an exponential improvement from the standard sampling model:
\begin{theorem}\label{theo:tester:monotonicity:intcond}
There exists an $\tildeO{\frac{\log^5{n}}{\eps^4}}$-query tester for monotonicity in the \ICOND model.
\end{theorem}

\noindent We now prove these two theorems, starting with~\autoref{theo:tester:monotonicity:intcond}. In doing so, we will also derive a weaker, $\poly(\log n, 1/\eps)$-query tester for \COND; before turning in~\autoref{sec:polyeps:cond:tester} to the constant-query tester of~\autoref{theo:tester:monotonicity:cond}.

\subsection{A \texorpdfstring{$\poly(\log n, 1/\eps)$}{poly(log n,1/eps)}-query tester for \ICOND}\label{sec:polyneps:cond:tester}
Our algorithm (\autoref{algo:general:monotonicity:algorithm}) follows the same overall idea as the one from \cite{BKR:04}, which a major difference. As in theirs, the first step will be to partition $[n]$ into a small number of intervals, such that the conditional distribution $\D_I$ on each interval $I$ is close to uniform; that is,
\begin{equation}
  \totalvardist{\D_I}{\uniform_I} = \sum_{i\in I} \abs{\frac{\D(i)}{\D(I)}-\frac{1}{\abs{I}}} \leq \frac{\eps}{4}\;.
\end{equation}
The original approach (in the sampling model) of Batu et al. was based on estimating the $\lp[2]$ norm of the conditional distribution \emph{via} the number of collisions from a sufficiently large sample; this yielded a $\tildeO{\sqrt{n}}$ sample complexity.

However, using directly as a subroutine (in the \COND model) an algorithm for (tolerantly) testing uniformity, one can perform this first step with $\ell_{\max}\log\frac{1}{\delta} = \ell_{\max}\log\ell_{\max}$ calls\footnote{Where the logarithmic dependence on $\delta$ aims at boosting the (constant) success probability of the uniformity testing algorithm, in order to apply a union bound over the $\bigO{\ell_{\max}}$ calls.} to this subroutine, each with approximation parameter $\frac{\eps}{4}$ (the proof of correctness from \cite{BKR:04} does not depend on how the test of uniformity is actually performed, in the partitioning step).\medskip 

\noindent A first idea would be to use for this the following result:
\begin{fact}[\cite{CRS:14}]\label{fact:crs:unif:testing}
  One can test \eps-uniformity of a distribution $\D_r$ over $[r]$ in the conditional sampling model:
    \begin{itemize}
      \item with $\tildeO{{1}/{\eps^2}}$ samples, given access to a $\COND_{\D_r}$ oracle;
      \item with $\tildeO{{\log^3 r}/{\eps^3}}$ samples, given access to a $\ICOND_{\D_r}$ oracle.
    \end{itemize}
\end{fact}
However, this does \emph{not} suffice for our purpose: indeed,~\autoref{algo:general:monotonicity:algorithm} needs in Step~\ref{algo:general:step:testunif} not only to reject distributions that are too far from uniform, \emph{but also to accept those that are close enough.} A standard uniformity tester as the one above does not ensure the latter condition: for this, one would \textit{a priori} need \emph{tolerant} tester for uniformity. While \cite{CRS:14} does describe such a tolerant tester (see~\autoref{lem:cond:estim}), it only applies to \COND{} -- and we aim at getting an \ICOND tester.

To resolve this issue, we observe that what the algorithm requires is slightly weaker: namely, to distinguish distributions on an interval $I$ that \textsf{(a)} are $\Omega(\eps)$-far from uniform from those that are \textsf{(b)} $O(\eps/\abs{I})$-close to uniform \emph{in $\lp[\infty]$ distance}. It is not hard to see that the two testers of~\autoref{fact:crs:unif:testing} can be adapted in a straightforward fashion to meet this guarantee, with the same query complexity. Indeed, \textsf{(b)} is equivalent to asking that the ratio $\D(x)/\D(y)$ of any two points in $I$ be in $[1-\eps,1+\eps]$, which is exactly what both testers check.

\begin{algorithm}[H]
  \begin{algorithmic}[1]
  \algblock[FirstStep]{PartitionStart}{PartitionEnd}
    \Require $\mathcal{O}\in\{\COND,\ICOND\}$ access to $\D$
    \State Define $\ell_{\max}\eqdef \bigO{\frac{\log^2 n}{\eps}}$, $\delta\eqdef\bigO{\frac{1}{\ell_{\max}}}$.
    \State Draw $m\eqdef\bigO{\frac{\eps}{\ell_{\max}}\log\frac{1}{\delta}}$ samples $h_1,\dots,h_m$.
    \PartitionStart
        \State Start with interval $I\gets[n]$
        \Repeat
          \State\label{algo:general:step:testunif} Test (with probability $\geq 1-\delta$) if $\D_I$ is $\eps/4$-close to the uniform distribution on $I$
          \If{ $\totalvardist{\D_I}{\uniform_I} > \frac{\eps}{4}$ }
             \State bisect $I$ in half             \State recursively test each half that contains at least one of the $h_i$'s, mark them as ``light'' otherwise \label{algo:general:monotonicity:step:recurse}
          \ElsIf{ $\ell_{\max}$ splits have been made }
            \State \Return \fail
          \EndIf
        \Until{all intervals are close to uniform or have been marked ``light''}
    \PartitionEnd
    \State Let $\mathcal{I}_\ell=\langle I_1,\dots,I_\ell\rangle$ denote the partition of $[n]$ into intervals induced by the leaves of the recursion from the previous step.
    \State Obtain an additional sample $T$ of size $\bigO{\frac{\log^4 n}{\eps^2}}$.
    \State Let $\hat{\D}$ denote the $\ell$-flat distribution described by $(\vect{w}, \mathcal{I}_\ell)$ where $\omega_j$ is the fraction of samples from $T$ falling in $I_j$.
    \If{$\hat{\D}$ is $(\eps/2)$-close to monotone} \Comment{Can be tested in $\poly(\ell)$-time (\cite[Lemma 8]{BKR:04})}
      \State \Return \accept
    \Else
         \State \Return \fail
    \EndIf 
  \end{algorithmic}
  \caption{\label{algo:general:monotonicity:algorithm}General algorithm \textsc{TestMonCond\textsuperscript{$\mathcal{O}$}}}
\end{algorithm}

\noindent As a corollary, we get:
\begin{corollary}
  Given access to a conditional oracle $\mathcal{O}$ for a distribution $\D$ over $[n]$, the algorithm $\textsc{TestMonCond}^{\mathcal{O}}$ outputs \accept when $\D$ is monotone and \fail when it is \eps-far from monotone, with probability at least 2/3. The algorithm uses
    \begin{itemize}
      \item $\tildeO{\frac{\ell_{\max}}{\eps} + \frac{\ell_{\max}}{\eps^2}+\frac{\log^4 n}{\eps^2}} = \tildeO{\frac{\log^2 n}{\eps^3}+\frac{\log^4 n}{\eps^2}}$ samples, when $\mathcal{O}=\COND_{\D}$;
      \item $\tildeO{\frac{\ell_{\max}}{\eps} + \ell_{\max}\frac{\log^3 n}{\eps^3}+\frac{\log^4 n}{\eps^2}} = \tildeO{\frac{\log^5 n}{\eps^4}}$ samples, when $\mathcal{O}=\ICOND_{\D}$.
    \end{itemize}
\end{corollary}
\noindent This in turn implies~\autoref{theo:tester:monotonicity:intcond}. Note that we make sure in Step~\ref{algo:general:monotonicity:step:recurse} that each of the intervals we recurse on contains at least one of the ``reference samples'' $h_i$: this is in order to guarantee all conditional queries made on a set with non-zero probability. Discarding the ``light intervals'' can be done without compromising the correctness, as with high probability each of them has probability weight at most $\frac{\eps}{4\ell_{\max}}$, and therefore in total the light intervals can amount to at most $\eps/4$ of the probability weight of $\D$ -- as in the original argument of Batu et al., we can still conclude that with high probability $\hat{\D}$ is \eps/2-close to $\D$.

\subsection{A \texorpdfstring{$\poly(1/\eps)$}{poly(1/eps)}-query tester for \COND}\label{sec:polyeps:cond:tester}

The idea in proving~\autoref{theo:tester:monotonicity:cond} is to reduce the task of testing monotonicity to another property, but on a (related) distribution \emph{over a much smaller domain}. We begin by introducing a few notations, and defining the said property:

\subsubsection{Reduction from testing properties over \texorpdfstring{$[\ell]$}{[l]}}

\noindent For fixed $\alpha$ and $\D$, let $\Dred_{\alpha}$ be the \emph{reduced} distribution on $[\ell]$ with respect to the oblivious decomposition $\mathcal{I}_\alpha$, where all throughout $\ell=\ell(\alpha,n)$ as per~\autoref{def:birge:obl:decomp}; i.e,
\[ \forall k\in[\ell],\quad \Dred_{\alpha}(k) = \D(I_k) = \Dflat{\alpha}(I_k) \]
Note that given oracle access $\SAMP_{D}$, it is easy to simulate $\SAMP_{\Dred_{\alpha}}$.

\begin{definition}[Exponential Property]\label{definition:exponential:property}
  Fix $n$, $\alpha$, and the corresponding $\ell=\ell(n,\alpha)$. For distributions over $[\ell]$, let the property $\property_\alpha$ be defined as 
  {\sf ``$Q\in\property_\alpha$ if and only if there exists $D\in\mathcal{M}$ over $[n]$ such that $Q=\Dred_{\alpha}$.''}
\end{definition}
\begin{fact}\label{fact:property:and:expanded:monotonicity}
  Given a distribution $Q$ over $[\ell]$, let $\operatorname{expand}_\alpha(Q)$ denote the distribution over $[n]$ obtained by ``spreading'' uniformly $Q(k)$ over $I_k$ (again, considering the oblivious decomposition of $[n]$ for $\alpha$). Then,
  \begin{equation}
    Q\in\property_\alpha \Leftrightarrow \operatorname{expand}_\alpha(Q)\in\mathcal{M}
  \end{equation}
\end{fact}
\begin{fact}\label{fact:property:and:exponential:factor}
  Given a distribution $Q$ over $[\ell]$, the following also holds:\footnote{We point out that the equivalence stated here once again ignores, for the sake of conceptual clarity, technical details arising from the discrete setting. Taking these into account would yield a slightly weaker characterization, with a twofold implication instead of an equivalence; which would still be good enough for our purpose.}
  \begin{equation}
    \property_\alpha(Q) \Leftrightarrow \forall k < \ell,\quad Q(k+1) \leq (1+\alpha)Q(k)
  \end{equation}
\end{fact}
\begin{remark}\label{remark:property:and:expanded:monotonicity}
  It follows from~\autoref{fact:property:and:expanded:monotonicity} that, for $D$ over $[n]$, 
    \begin{equation}
      \Dflat\alpha\in\mathcal{M}  \Leftrightarrow \Dred_\alpha\in\property_\alpha
    \end{equation}
\end{remark}
We shall also use the following result on flat distributions (adapted from \cite[Lemma 7]{BKR:04}):  
  \begin{fact}\label{fact:flat:close:monotone:iff:close:flat:monotone}
    $\Dflat{\alpha}$ is \eps-close to monotone if and only if it is \eps-close to a $\mathcal{I}_\alpha$-flat monotone distribution (that is, a monotone distribution piecewise constant, according to the same partition $\mathcal{I}_\alpha$).
  \end{fact}
  \begin{proof}
    The sufficient condition is trivial; for the necessary one, assume $\Dflat{\gamma}$ is \eps-close to monotone, and let $Q$ be a monotone distribution proving it. We show that \mbox{$\totalvardist{\Dflat{\gamma}}{\birge[Q]{\gamma}} \leq \eps$}:
    \begin{align*}
      2\totalvardist{\Dflat{\gamma}}{\birge[Q]{\gamma}} &= \sum_{k=1}^\ell\abs{\Dflat{\gamma}(I_k) - \birge[Q]{\gamma}(I_k)}
      = \sum_{k=1}^\ell\abs{\Dflat{\gamma}(I_k) - Q(I_k)} \\
      &= \sum_{k=1}^\ell\abs{\sum_{i\in I_k} \left( \Dflat{\gamma}(i) - Q(i) \right) }
      \leq \sum_{k=1}^\ell \sum_{i\in I_k} \abs{ \Dflat{\gamma}(i) - Q(i) } \\
      &= \sum_{i=1}^n \abs{ \Dflat{\gamma}(i) - Q(i) } = 2\totalvardist{\Dflat{\gamma}}{Q} \\
      &\leq 2\eps.
    \end{align*}
  \end{proof}
  
\noindent Observe that~\autoref{fact:flat:close:monotone:iff:close:flat:monotone},~\autoref{remark:property:and:expanded:monotonicity} and~\autoref{fact:property:and:expanded:monotonicity} altogether imply that, for $\mathcal{I}_\alpha$-flat distributions, distance to monotonicity and distance to $\property_\alpha$ of the reduced distribution are equal.

\subsubsection{Efficient approximation of distance to \texorpdfstring{$\Dflat{}$}{Dflat}}
\begin{lemma}\label{lemma:cond:estimate:distance:flattening}
Given \COND access to a distribution $\D$ over $[n]$, there exists an algorithm that, on input $\alpha$ and $\eps,\delta\in(0,1]$, makes $\tildeO{\frac{1}{\eps^{22}}\log\frac{1}{\delta}}$ queries (independent of $\alpha$) and outputs a value $\hat{d}$ such that, with probability at least $1-\delta$, $\abs{\hat{d} - \totalvardist{\D}{\Dflat{\alpha}}} \leq \eps$.
\end{lemma}
\begin{proof}
We describe such algorithm for a constant probability of success; boosting the success probability to $1-\delta$ at the price of a multiplicative $\log\frac{1}{\delta}$ factor can then be achieved by standard techniques (repetition, and taking the median value). Let $\D$, \eps and $\mathcal{I}_{\alpha}$ be defined as before; define $Z$ to be a random variable taking values in $[0,1]$, such that, for $k\in[\ell]$, $Z$ is equal to $\totalvardist{\D_{I_k}}{\uniform_{I_k}}$ with probability $\omega_k=\D(I_k)$. It follows that
\begin{align}
  \shortexpect Z &= \sum_{k=1}^\ell \omega_k \totalvardist{\D_{I_k}}{\uniform_{I_k}} = \half\sum_{k=1}^\ell \omega_k \sum_{i\in I_k} \abs{\D_{I_k}(i)-\frac{1}{\abs{I_k}}} \notag \\
  &= \half\sum_{k=1}^\ell \sum_{i\in I_k} \abs{\D(i)-\frac{\D(I_k)}{\abs{I_k}}} = \half\sum_{i=1}^n \abs{\D(i)-\Dflat{\alpha}} \notag \\
  &= \totalvardist{\D}{\Dflat{\alpha}}.
\end{align}
Putting aside for now the fact that we only have (using as a subroutine the \COND algorithm from~\autoref{lem:cond:estim} to estimate the distance to uniformity) access to \emph{additive approximations} of the $\totalvardist{\D_{I_k}}{\uniform_{I_k}}$'s, one can simulate independent draws from $Z$ by taking each time a fresh sample $i\sim \D$, looking up the $k$ for which $i\in I_k$, and calling the \COND subroutine to get the corresponding value. Applying a Chernoff bound, only $\bigO{1/\eps^2}$ such draws are needed, each of them costing $\tildeO{1/\eps^{20}}$ \COND queries.

\paragraph{Dealing with approximation.} It suffices to estimate $\shortexpect Z$ within an additive $\eps/2$, which can be done with probability 9/10 by simulating $m=\bigO{1/\eps^2}$ samples from $Z$. To get each sample, for the index $k$ drawn we can call the \COND subroutine with parameters $\eps/2$ and $\delta=1/(10m)$ to obtain an estimate of $\totalvardist{\D_{I_k}}{\uniform_{I_k}}$. By a union bound we get that, with probability at least 9/10, all estimates are within an additive $\eps/2$ of the true value, incurring only a $\bigO{\log 1/\eps}$ additional factor in the overall sample complexity $\tildeO{1/\eps^{20}}$. Conditioned on this, we get that the approximate value we compute instead of $\shortexpect Z$ is off by at most $\eps/2+\eps/2=\eps$ (where the first term corresponds to the approximation of the value of $Z$ for each draw, and the second comes from the additive approximation of $\shortexpect Z$ by sampling).
\end{proof}

\subsubsection{The algorithm}

\begin{algorithm}[H]
  \begin{algorithmic}[1]
    \Require \COND access to $\D$
    \State\label{algo:cond:2:step:1} Simulating $\COND_{\Dred_{\alpha}}$, check if $\Dflat{\alpha}$ is $(\eps/4)$-close to monotone by testing $(\eps/4)$-farness (of $\Dred_{\alpha}$) to $\property_{\alpha}$; \Return \fail if not.
    \State\label{algo:cond:2:step:2} Test whether $\Dflat{\alpha}$ is $(\eps/4)$-close to $\D$ using the sampling approach discussed above; \Return \fail if not.
    \State\Return\accept
  \end{algorithmic}
  \caption{\label{algo:cond:monotonicity:algorithm}Algorithm \textsc{TestMonCond}}
\end{algorithm}
The tester is described in~\autoref{algo:cond:monotonicity:algorithm}. The second step, as argued in~\autoref{lemma:cond:estimate:distance:flattening}, uses $\tildeO{{1}/{\eps^{22}}}$ samples; we will show in~\autoref{ssec:testing:dist:to:exp:property} that \emph{efficiently testing \eps-farness to $\property_\gamma$} is also achievable with $\tildeO{{1}/{\eps^6}}$ \COND queries -- concluding the proof of~\autoref{theo:tester:monotonicity:cond}.

\paragraph{Correctness of~\autoref{algo:cond:monotonicity:algorithm}.}
Assume we can efficiently perform the two steps, and condition on their execution being correct (as each of them is run with for instance parameter $\delta=1/10$, this happens with probability at least $3/4$).
\begin{itemize}
  \item If $\D$ is monotone non-increasing, so is $\Dflat{\alpha}$; by~\autoref{remark:property:and:expanded:monotonicity}, this means that $\property_{\alpha}(\Dred_{\alpha})$ holds, and the first step passes.~\autoref{theorem:birge:obl:decomp} then ensures that $D$ and $\Dflat{\alpha}$ are $\alpha$-close, and the algorithm outputs \accept;
  \item If $\D$ is \eps-far from monotone, then either \textsf{(a)} $\Dflat{\alpha}$ is $\frac{\eps}{2}$-far from monotone or \textsf{(b)} $\totalvardist{\D}{\Dflat{\alpha}} > \frac{\eps}{2}$; if \textsf{(b)} holds, no matter how the algorithm behaves in first step, the algorithm not go further that the second step, and output \fail. Assume now that \textsf{(b)} does not hold, i.e. only \textsf{(a)} is satisfied. By putting together~\autoref{fact:flat:close:monotone:iff:close:flat:monotone},~\autoref{remark:property:and:expanded:monotonicity} and~\autoref{fact:property:and:expanded:monotonicity}, we conclude that \textsf{(a)} implies that $\Dred_{\alpha}$ is $\frac{\eps}{2}$-far from $\property_{\alpha}$, and the algorithm outputs \fail in the first step.
\end{itemize}

\subsubsection{Testing \texorpdfstring{\eps-farness to $\property_\gamma$}{eps-farness to P}}\label{ssec:testing:dist:to:exp:property}

To achieve this objective, we begin with the following lemma, which relates the distance between a distribution $Q$ and $\property_\alpha$ to the total weight of points that violate the property.
\begin{lemma}\label{lemma:distance:palpha::weight:witnesses}
Let $Q$ be a probability distribution over $[\ell]$, and $W = \setOfSuchThat{ i }{ Q(i) > (1+\alpha)Q(i-1) }$ be the set of \emph{witnesses} (points which violate the property). Then, the distance from $Q$ to the property $\property_\alpha$ is $\bigO{1/\alpha} Q(W)$.
\end{lemma}
\begin{proof}
  One can define the procedure $\textsc{Fixup}_\alpha$ which, given a distribution $Q$ and the corresponding $W$, acts as follows:
  \begin{algorithmic}
    \Ensure $\textsc{Fixup}_\alpha(Q)$ is a distribution satisfying $\property_\alpha$
    \State $Q^\prime \leftarrow Q$, $W^\prime \leftarrow W$
    \While{$W^\prime\neq\emptyset$}
      \State Let $i > 1$ be the smallest (leftmost) point in $W^\prime$, set $\Delta\leftarrow 0$ and $d\leftarrow 1$.
      \While{$Q(i-d) \geq \frac{Q^\prime(i)}{(1+\alpha)^d}$} \Comment{Increase the weight of the predecessors of $i$}
        \State $\Delta\leftarrow \Delta + \left(Q^\prime(i)(1+\alpha)^{-d}-Q^\prime(i-d)\right)$
        \State $Q^\prime(i-d) \leftarrow Q^\prime(i)(1+\alpha)^{-d}$
      \EndWhile
      \While{$\Delta>0$}  \Comment{Remove this weight from the rightmost points}
        \State $Q^\prime(k) \leftarrow Q^\prime(k) - \min(\Delta, Q^\prime(k))$
        \State $\Delta\leftarrow \Delta - \min(\Delta, Q^\prime(k))$
      \EndWhile
    \EndWhile
    \State\Return $Q^\prime$
  \end{algorithmic}
   If $\bar{\Delta}$ denotes the total probability weight reassigned (i.e, the sum of the $\bar{\Delta}_i$'s, where $\bar{\Delta}_i$ is the total weight reassigned for witness $i$), then we have that $2\totalvardist{Q}{\textsc{Fixup}_\alpha(Q)} = \sum_{i=1}^\ell\abs{\textsc{Fixup}_\alpha(Q)(i)-Q(i)}\leq \sum_{i\in W} 2\bar{\Delta}_i = 2\bar{\Delta}$; and since 
   \[ \bar{\Delta} \leq \sum_{i \in W} Q(i)\cdot\left( \frac{1}{1+\alpha} + \frac{1}{(1+\alpha)^2} + \dots \right) \leq Q(W)\cdot\frac{1+\alpha}{\alpha}\]
   we get that $\totalvardist{Q}{\textsc{Fixup}_\alpha(Q)}= \bigO{\frac{1}{\alpha}} Q(W)$ (and $\textsc{Fixup}_\alpha(Q)$ clearly satisfies $\property_\alpha$).
\end{proof}

\begin{remark}~\autoref{lemma:distance:palpha::weight:witnesses} implies that when $Q$ is \eps-far from having the property, it suffices to sample $\bigO{\frac{1}{\alpha\eps}}$ points according to $Q$ and compare them to their neighbors to detect a violation with high probability. Note that this last test would be easy, granted access to an \emph{exact} \EVAL oracle; for the purpose of this section, however, we can only use an approximate one. The lemma below addresses this issue, by ensuring that there will be many points ``patently'' violating the property.
\end{remark}

\begin{lemma}\label{lemma:distance:palpha::weight:tau:witnesses}
Let $Q$ be as above, and, for $\tau >0$, let $W_\tau = \setOfSuchThat{ i }{ Q(i) > (1+\alpha+\tau)Q(i-1) }$ be the set of $\tau$-\emph{witnesses} (so that $W=\bigcup_{\tau > 0} W_\tau)$. Then, the distance from $Q$ to the property $\property_\alpha$ is at most $\bigO{1/(\alpha+\tau)} Q(W_\tau)+\bigO{\tau/\alpha^2}$.
\end{lemma}
\begin{corollary}\label{cor:distance:palpha::weight:tau:witnesses}
  Taking $\alpha=\bigTheta{\eps}$ and $\tau=\eps\alpha^2$, we get that if $Q(W_\tau)\leq\eps^2$, then $Q$ is $\bigO{\eps}$-close to $\property_\alpha$.
  \end{corollary}

\begin{proof}[Proof of~\autoref{lemma:distance:palpha::weight:tau:witnesses}]
We first apply the ``fix-up'' as defined in the proof of~\autoref{lemma:distance:palpha::weight:witnesses} to get $Q'$ such that $Q'(i) \leq (1+\alpha+\tau)Q'(i-1)$ for all $i$, at a cost of $\bigO{\frac{1}{\alpha+\tau}}Q(W_\tau)$. Next, we obtain a distribution $Q''$ satisfying $\property_\alpha$ by apply the fix-up to all $i$  such that $Q'(i) > (1+\alpha)Q'(i-1)$.
If we start from some violating $i$ (until  we reach some $k_i= i-d$ such that $Q'(k)$ does not need to be fixed since $Q''(k+1) \leq (1+\alpha)Q'(k)$), we know that before the fix-up, for each $1\leq d\leq k_i$, 
$Q'(i-d) \geq  \frac{Q'(i)}{(1+\alpha+\tau)^d}$, and now, after the fix-up $Q''(i) = Q'(i)$ and $Q''(i-d) = \frac{Q'(i)}{(1+\alpha)^d}$. 
The cost of this increase is:
\begin{equation}
  Q''(i-d) - Q'(i-d) \leq   Q'(i)\cdot \left(\frac{1}{(1+\alpha)^d}  - \frac{1}{(1+\alpha + \tau)^d}\right)\;
\end{equation}
Using the fact that 
$(1+\alpha + \tau) = (1+\alpha)(1+\tau/(1+\alpha)) < (1+\alpha)(1+\tau)$ so that $\frac{1}{1+\alpha+\tau} > \frac{1}{(1+\alpha)(1+\tau)}$,
we get
\begin{align*}
Q''(i-d) - Q'(i-d) 
  &\leq  Q'(i)\cdot\frac{1}{(1+\alpha)^d}\cdot \left(1 - \frac{1}{(1+\tau)^d}\right) \\
&= Q'(i)\cdot\frac{1}{(1+\alpha)^d}\cdot \frac{(1+\tau)^d -1}{(1+\tau)^d} \\
&\leq Q'(i)\cdot\frac{\tau d}{(1+\alpha)^d}
\end{align*}
(where the last inequality uses the fact that $(1+\tau)^d -1 = d\tau + \binom{d}{2}\tau^2 +\dots + d\tau^{d-1}+ \tau^d$
which is less than $d\tau\cdot (1+\tau)^d = d\tau + d^2\tau^2 + d\binom{d}{2}\tau^3 + \dots + d^2\tau^d + d\tau^{d+1}$).
Since, for $x\in[0,1)$
\begin{equation}
\sum_{i=1}^n i\cdot x^i \leq \sum_{i=1}^\infty i\cdot x^i = \frac{x}{(1-x)^2}\;,
\end{equation}
we get that
\begin{equation}
\sum_{1\leq d\leq k_i} (Q''(i-d)-Q'(i-d)) \leq Q'(i)\cdot \tau \sum_{d=1}^\infty d\cdot \frac{1}{(1+\alpha)^d}
   < Q'(i)\cdot\frac{\tau(1+\alpha)}{\alpha^2}\;.
\end{equation}
By summing over all $i$ from which we start the (second) ``fix-up,'' we get an increase of at most $\frac{\tau(1+\alpha)}{\alpha^2}$. By the triangle inequality, the total distance from $Q$ to $Q''$ is therefore at most
\begin{equation}
  \frac{1+\alpha+\tau}{\alpha+\tau}Q(W_\tau) + \frac{\tau(1+\alpha)}{\alpha^2}\;.
\end{equation}
\end{proof}

  By leveraging~\autoref{cor:distance:palpha::weight:tau:witnesses}, we are able to obtain efficient approximation of the distance of a distribution to the ``exponential property'':
  \begin{theorem}
    There exists a constant $0<c<1$ such that, for any $\eps>0$: if $Q$ satisfies $\property_{\alpha}$ (where $\alpha=c\eps$),then with probability at least 2/3 Algorithm \textsc{TestingExponentialProperty} returns \accept, and if $Q$ is $\bigOmega{\eps}$-far from $\property_{\alpha}$, then with probability at least 2/3 Algorithm \textsc{TestingExponentialProperty} returns \fail. The number of \PCOND queries performed by the algorithm is $\tildeO{\frac{1}{\eps^8}}$.
  \end{theorem}
  \begin{proof}[Proof (sketch)] The algorithm can be found in~\autoref{algo:testing:palpha}. We here prove its correctness, before turning to its sample complexity.
  \paragraph{Correctness.}
  Conditioning on the events of all calls to \textsc{Compare} returning a correct value (by a union bound, this happens with probability at least $9/10$), we have that:
  \begin{itemize}
    \item if $Q$ satisfies $\property_\alpha$, then for any sample $s_i>1$, \textsc{Compare} can only return \textsf{Low} or a value $\rho$. In the latter case, since $s_i\notin W=\emptyset$, it holds that $Q(s_i-1)\leq(1+\alpha)Q(s_i)$, and therefore $\rho \geq (1-\eta)\frac{Q(s_i)}{Q(s_i-1)} \geq \frac{1-\eta}{1+\alpha} > \frac{1+\eta}{1+\alpha+\tau}$ (where the last inequality holds because of the choice of $\eta$), and the algorithm does not reject;
    \item if however $Q$ is $\bigOmega{\eps}$-far from $\property_\alpha$,~\autoref{cor:distance:palpha::weight:tau:witnesses} ensures that with probability at least $9/10$ one of the samples will belong to $W_\tau$. For such a $s_i$, \textsc{Compare} will either return \textsf{High} (and the algorithm will reject) or a value $\rho$. In the latter case, it will be the case that $Q(s_i-1)>(1+\alpha+\tau)Q(s_i)$, and thus $\rho < (1+\eta)\frac{Q(s_i)}{Q(s_i-1)} \geq \frac{1+\eta}{1+\alpha+\tau}$, and the algorithm will reject.
  \end{itemize}
  The outcome of the algorithm will hence be correct with probability at least $3/4$.
  
  \paragraph{Sample complexity.} By choice of $\alpha$, $m=\bigTheta{1/\eps^2}$ and $\tau=\bigTheta{1/\eps^3}$;  each of the $m$ calls to \textsc{Compare} costs $\bigO{\frac{K\log 1/\delta}{\eta^2}}=\bigO{\frac{\log m}{\tau^2}}=\tildeO{\frac{1}{\eps^6}}$.
  \end{proof}

  \begin{algorithm}
  \begin{algorithmic}
    \Require $\PCOND$ access to $Q$, $\alpha\in[0,1)$ \Comment{Useful for $\alpha=\bigTheta{\eps} < 1$}
    \Ensure with probability at least $3/4$ returns \fail if $Q$ is $\bigO{\eps}$-close to $\property_\alpha$, and \accept if it satisfies $\property_\alpha$.
    \State Set $\tau\eqdef\eps\alpha^2$
    \State Draw $m\eqdef\bigTheta{\frac{1}{\eps\alpha}}$ samples $s_1,\dots,s_m$ from $Q$  \Comment{Contains an element from $W_\tau$ w.h.p.}
    \For{$i=1 \textbf{ to } m$}
      \If{$s_i \geq 2$}
        \State Call \textsc{Compare} (from~\autoref{lem:cond:compare}) on $\{s_{i}-1\}$, $\{s_i\}$ with $\eta=\frac{\tau}{2}$, $K=2$ and $\delta=\frac{1}{10m}$.
        \If{ the procedure outputs \textsf{High} }           \Return \fail
        \ElsIf{it outputs a value $\rho$} \Comment{$\frac{1-\eta}{\rho}\cdot Q(s_i) \leq Q(s_i-1) \leq \frac{1+\eta}{\rho}\cdot Q(s_i)$}
          \If{$\rho < \frac{1+\eta}{1+\alpha+\tau}$} \Return \fail \EndIf
        \EndIf
      \EndIf
    \EndFor
    \State\Return \accept
  \end{algorithmic}
    \caption{\label{algo:testing:palpha}\textsc{TestingExponentialProperty}}
  \end{algorithm}
 
\section{With \EVAL access}\label{sec:eval}
\makeatletter{}In this section, we describe a $\poly(\log n, 1/\eps)$-query tester for monotonicity in the Evaluation Query model (\EVAL), in which the testing algorithm is granted query access to the probability mass function unknown distribution -- but not the ability to sample from it.

\begin{remark}[On the relation to {$\lp[p]$}-testing for functions on the line]
We observe that the results of Berman et al. \cite{BRY:14} in testing monotonicity of functions with relation to $\lp[p]$ distances do not directly apply here. Indeed, while their work is indeed concerned with functions $f\colon[n]\to[0,1]$ to which query access is granted, two main differences prevent us from using their techniques for \EVAL access to distributions: first, the distance they consider is normalized, by a factor $n$ in the case of $\lp[1]$ distance. A straightforward application of their result would therefore imply replacing $\eps$ by $\eps^\prime = \eps/n$ in their statements, incurring a prohibitive sample complexity. Furthermore, even adapting their techniques and structural lemmata is not straightforward, as distance to monotone $[0,1]$-valued \emph{functions} is not directly related to distance to monotone distributions: specifically, the main tool leveraged in their reduction to Boolean Hamming testing (\cite[Lemma 2.1]{BRY:14}) does no longer hold for distributions.
\end{remark}

\subsection{A \texorpdfstring{$\poly(\log n, 1/\eps)$}{poly(log n, 1/eps)}-query tester for \EVAL}\label{sec:polylogneps:eval:ub}

We start by stating two results we shall use as subroutines, before stating and proving our theorem.

\begin{lemma}\label{lemma:eval:learn:monotone}
Given \EVAL access to a \emph{monotone} distribution $\D$ over $[n]$, there exists a (non-adaptive) algorithm\footnotemark that, on input $\eps$, makes $\bigO{\frac{\log n}{\eps}}$ queries and outputs a monotone distribution $\hat{\D}$ such that $\totalvardist{\hat{\D}}{\D} \leq \eps$. Furthermore, $\hat{\D}$ is an $\bigO{\frac{\log n}{\eps}}$-histogram.
\end{lemma}
\footnotetext{Recall that a non-adaptive tester is an algorithm whose queries do not depend on the answers obtained from previous ones, but only on its internal randomness. Equivalently, it is a tester that can commit ``upfront'' to all the queries it will make to the oracle.}
\begin{proof}
  This follows from adapting the proof of \autoref{theorem:birge:obl:decomp} as follows: we consider the same oblivious partition of $[n]$ in $\ell=\bigO{\log n/\eps}$ intervals, but instead of taking as in \eqref{eq:birge:def:dflat} the weight of a point $i\in I_k$ to be the average $\D(I_k)/\abs{I_k}$, we consider the average of the \emph{endpoints} of $I_k=(a_k, a_{k+1}]$:
  \[
    \forall k\in [\ell], \forall i\in I_k,\quad \tilde{\D}_{\eps}(i) = \frac{\D(a_k)+\D(a_{k+1})}{2}\;.
  \]
Clearly, this hypothesis can be (exactly) computed by making $\ell$ \EVAL queries. The result directly follows from observing that, in the proof of his theorem, Birg\'e first upperbounds $\normone{\birge[\D]{\eps} - \D}$ by $\normone{\tilde{\D}_{\eps}-\D}$, before showing the latter -- which is the quantity we are interested in -- is at most 2\eps (see \cite[Eq. (2.4)--(2.5)]{Birge:87}). \new{The last step to be taken care of is the fact that $\tilde{\D}_{\eps}$, as defined, might not be a distribution -- i.e., it may not sum to one. But as $\tilde{\D}_{\eps}$ is fully known, it is possible to efficiently (and without taking any additional sample) compute the $\ell$-histogram monotone \emph{distribution} $\widehat{\D}_{\eps}$ which is closest to it. We are guaranteed that $\widehat{\D}_{\eps}$ will be at most $4\eps$-far from $\tilde{\D}_{\eps}$ in $\lp[1]$ distance, as there exists one particular distribution, namely $\birge[\D]{\eps}$, that is (being at a distance at most $2\eps$ of $\D$ as well). Therefore, overall $\widehat{\D}_{\eps}$ is a monotone distribution that is at most $6\eps$-far from $\D$ in $\lp[1]$ distance, i.e. $\totalvardist{\D}{\widehat{\D}_{\eps}} \leq 3\eps$.}
\end{proof}

\todonote[inline, color=cyan!30]{Could we make this variant of Birg\'e also robust? Not straightforward -- for a start, we cannot argue that $\totalvardist{\widehat{P}_{\alpha}}{\widehat{\D}_{\alpha}} \leq \eps$, as we had for $\totalvardist{\birge[P]{\alpha}}{\birge[\D]{\alpha}}$ in \autoref{coro:birge:decomposition:robust}. And indeed the approach is very brittle -- e.g., if a lot of error is concentrated on the endpoints of the intervals. However, if we could, the approach here would give tolerant testing.}

\begin{theorem}[Tolerant identity testing ({\cite[Remark 3 and Corollary 1]{CR:14}})]\label{lemma:estimate:tolerant:identity}
Given \EVAL access to a distribution $\D$ over $[n]$, there exists a (non-adaptive) algorithm that, on input $\eps,\delta\in(0,1]$ and the full specification of a distribution $\D^\ast$, makes $\bigO{\frac{1}{\eps^2}\log\frac{1}{\delta}}$ queries and outputs a value $\hat{d}$ such that, with probability at least $1-\delta$, $\abs{\hat{d} - \totalvardist{\D}{\D^\ast}} \leq \eps$.
\end{theorem}

\begin{theorem}\label{theo:tester:monotonicity:eval}
There exists an $\bigO{\max\mleft( \frac{\log n}{\eps}, \frac{1}{\eps^2}\mright)}$-query tester for monotonicity in the \EVAL model.
\end{theorem}
\begin{proof}
Calling \autoref{lemma:eval:learn:monotone} with accuracy parameter $\eps/4$ enables us to learn a histogram $\widehat{\D}$ guaranteed, if $\D$ is monotone, to be $(\eps/4)$-close to $\D$; this by  making $\ell\eqdef\bigO{\frac{\log n}{\eps}}$ queries. Using \autoref{lemma:estimate:tolerant:identity}, we can also get with $\bigO{1/\eps^2}$ queries an estimate $\hat{d}$ of $\totalvardist{\D}{\widehat{\D}}$, accurate up to an additive $\eps/4$ with probability at least $2/3$. Combining the two, we get, with probability at least $2/3$,
\begin{enumerate}[(i)]
  \item $\widehat{\D}$, $(\eps/4)$-close to $\D$ if $\D$ is monotone;
  \item\label{item:eval:estimate} $\hat{d}\in [\totalvardist{\D}{\widehat{\D}} - \eps/4, \totalvardist{\D}{\widehat{\D}} + \eps/4]$.
\end{enumerate}\medskip

\noindent We can now describe the testing algorithm:
\begin{algorithm}[H]
  \begin{algorithmic}[1]
    \Require \EVAL access to $\D$
    \State\label{algo:eval:step:0} Set $\alpha\eqdef\eps/4$, and compute $\mathcal{I}_\alpha$.
    \State\label{algo:eval:step:1} Get a candidate approximation of $\D$ and test it for monotonicity, by:
      \begin{enumerate}[\sf(a)]
        \item Applying \autoref{lemma:eval:learn:monotone} with parameter $\alpha$ to obtain $\widehat{\D}$, histogram on $\mathcal{I}_\alpha$;
        \item\label{algo:eval:step:1:lp} Checking (offline) whether $\widehat{\D}$ is $(\eps/4)$-close to monotone; \Return \fail if not.       \end{enumerate}
    \State\label{algo:eval:step:2} Get an estimate $\hat{d}$ of $\totalvardist{\D}{\widehat{\D}}$ up to additive $\eps/4$, as per \ref{item:eval:estimate}; \Return \fail if $\hat{d} > \eps/2$.
    \State\Return\accept
  \end{algorithmic}
  \caption{\label{algo:eval:monotonicity:algorithm}Algorithm \textsc{TestMonEval}}
\end{algorithm}

To argue correctness, it suffices to observe that, conditioning on the estimate being as accurate as required (which happens with probability at least $2/3$):
\begin{itemize}
  \item if $\D\in\mathcal{M}$, then $\widehat{\D}\in\mathcal{M}$ as well and we pass the first step. We also know by \autoref{lemma:eval:learn:monotone} that in this case $\totalvardist{\D}{\widehat{\D}}\leq \alpha$, so that our estimate satisfies $\hat{d}\leq \alpha+\eps/4 = \eps/2$. Therefore, the algorithm does not reject here either, and eventually outputs \accept.
  \item conversely, if the algorithm outputs \accept, then we have both that \textsf{(a)} the distance of $\widehat{\D}$ to $\mathcal{M}$ is at most $\eps/4$, and \textsf{(b)} $\totalvardist{\D}{\widehat{\D}}\leq \hat{d} + \eps/4 \leq 3\eps/4$; so overall $\totalvardist{\D}{\mathcal{M}}\leq \eps$.
\end{itemize}

As for the query complexity, it is straightforward from the setting of $\alpha=\bigTheta{\eps}$ and the foregoing discussion (recall that Step~\ref{algo:eval:step:1:lp} can be performed efficiently, e.g. via linear programming (\cite[Lemma 8]{BKR:04})).
\end{proof}

\subsection{An \texorpdfstring{$\bigOmega{\log n}$}{Omega(log n)} (non-adaptive) lower bound for \EVAL}\label{sec:logn:eval:lb}
In this section, we show that, when focusing on \emph{non-adaptive} testers, \autoref{theo:tester:monotonicity:eval} is tight (note that the tester described in the previous section is, indeed, non-adaptive).
\begin{theorem}\label{theo:tester:monotonicity:eval:lb}
For any $\eps\in(0,1/2)$, any non-adaptive \eps-tester for monotonicity in the \EVAL model must perform $\frac{1}{4}\frac{\log n}{\eps}$ queries.
\end{theorem}
\makeatletter{}\begin{proof}\ We hereafter assume without loss of generality that $n/3$ is a power of two; and shall define a distribution over pairs of distributions, $\distrD$, such that the following holds. A random pair of distributions $(\D_1,\D_2)$ drawn from $\distrD$ will have $\D_1$ monotone, but $\D_2$ ${\eps}$-far from monotone. Yet, no non-adaptive deterministic testing algorithm can distinguish with probability $2/3$ (over the draw of the distributions) between $\D_1$ and $\D_2$, unless it performs $c\log n$ \EVAL queries. By Yao's minimax principle, this will guarantee that any non-adaptive randomized tester must perform at least $c\log n$ queries in the worst case.\medskip

More specifically, a pair of distributions $(\D_1,\D_2)$ is generated as follows. A parameter $m$ is chosen uniformly at random in the set 
\[
M\eqdef \mleft\{ \frac{2}{\kappa\eps}, \frac{2}{\kappa\eps}(1+\kappa), \dots, \frac{2}{\kappa\eps}(1+\kappa\eps)^k,\dots, \frac{n}{3} \mright\}\;,
\]
where $\kappa\eqdef\frac{4}{1-2\eps}$. $\D_1$ is then set to be the uniform distribution on $\{1,\dots,(2+\kappa\eps)m\}$; as for $\D_2$, it is defined as the histogram putting weight:
\begin{itemize}
  \item  $\frac{1}{2}-\eps$ on $\{1,\dots, m\}$;
  \item 0 on $I_m\eqdef \{m+1,\dots, \flr{(1+\frac{\kappa\eps}{2})m}\}$;
  \item $2\eps$ on  $J_m\eqdef \{\flr{(1+\frac{\kappa\eps}{2})m}+1,\dots, \flr{(1+\kappa\eps)m}\}$;
  \item and $\frac{1}{2}-\eps$ on $\{\flr{(1+\kappa\eps)m}+1,\dots, \flr{(2+\kappa\eps)m}\}$.
\end{itemize}
It is not hard to see that $\D_1$ is indeed monotone, and that the distance of $\D_2$ from monotone is exactly $\eps$.

\begin{figure}[!ht]\centering
  \begin{tikzpicture}[x=4pt, y=2pt, scale=0.7]
  \pgfmathsetmacro{\xsupport}{60}
  \pgfmathsetmacro{\xmax}{100}
  \pgfmathsetmacro{\ymax}{100}
  
    \draw[ultra thin, help lines, dotted] (0,0) grid (\xmax,\ymax);
  \draw [<->] (0,\ymax) node[above] {$D_j(i)$} -- (0,0) -- (\xmax,0)  node[right] {$i$};
    
  \draw[thin]  (0,{\ymax/3}) -- ({\xsupport/3},{\ymax/3}) -- ({\xsupport/3},0);
  \draw[thin]  ({\xsupport/2},0) -- ({\xsupport/2},{2*\ymax/3}) -- ({2*\xsupport/3},{2*\ymax/3}) -- ({2*\xsupport/3},{\ymax/3});
  \draw[dashed]  ({\xsupport/3},{\ymax/3}) -- ({2*\xsupport/3},{\ymax/3});
  \draw[thin]  ({2*\xsupport/3},{\ymax/3}) -- ({\xsupport},{\ymax/3})-- (\xsupport,0);
  
  \node[below] at ({\xsupport/3},0) {\scriptsize $\vphantom{\frac{3}{2}}m$};
  \node[below] at ({2*\xsupport/3},0) {\scriptsize $\vphantom{\frac{3}{2}}{(1+\kappa\eps)m}$};
  \node[below] at ({\xsupport},0) {\scriptsize $\vphantom{\frac{3}{2}}{(2+\kappa\eps)m}$};
  
  \foreach \i in {1,...,3}
     		\draw ({\i*\xsupport/3},1pt) -- ({\i*\xsupport/3},-3pt);

  \end{tikzpicture}\caption{\label{fig:eval:nonadapt:lb}Construction of $\D_1$ (dotted) and $\D_2$.}
\end{figure}
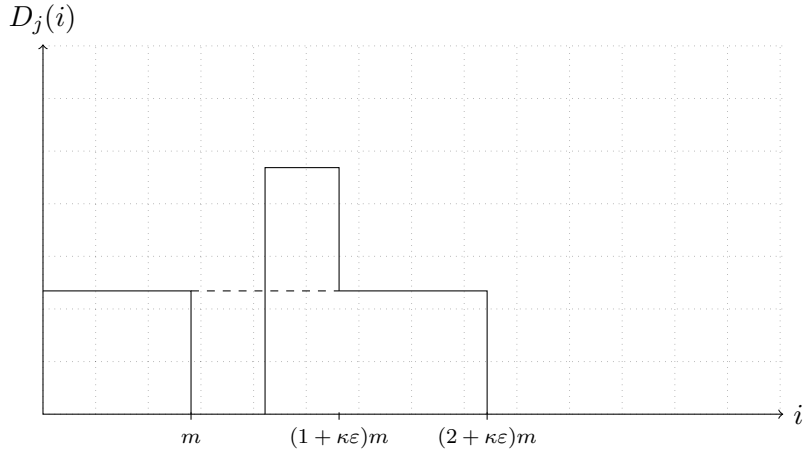

The key of the argument is to observe that if too few queries are made, then with high probability over the choice of $m$ no queries will hit the interval $I_m\cup J_m$; and that conditioning on this, what the tester sees in the \yes- and \no-cases is indistinguishable.

\begin{claim}
Let  \Tester be a deterministic, non-adaptive algorithm making $q \leq \frac{\log n}{4\eps}$ queries to the \EVAL oracle. Then, the probability (over the choice of $m$) that a query hits $I_m\cup J_m$ is less than $1/3$.
\end{claim}
\begin{proof}
  This follows from observing that the probability that any fixed point $x\in[n]$ belongs to $I_m\cup J_m$ is at most $\frac{1}{\abs{M}}=(1+\littleO{1})\frac{\eps}{\log n}$, as this can only happen for at most one value of $m$ (among $\abs{M}$ equiprobable choices). By Markov's inequality, this implies that the probability of any query falling into $I_m$ is at most $\frac{q}{\abs{M}} < \frac{1}{3}$ (for $n$ big enough).
\end{proof}
To see why the claims directly yields the theorem, observe that the above implies that for any such algorithm, $\abs{ \probaDistrOf{\D_1}{\Tester^{\D_1}=\yes} - \probaDistrOf{\D_2}{\Tester^{\D_2}=\yes} } < \frac{1}{3}$. But then \Tester cannot be a successful monotonicity tester, as otherwise it would accept $\D_1$ with probability at least $2/3$, and $\D_2$ with probability at most $1/3$.
\end{proof}

\subsection{An \texorpdfstring{$\tildeOmega{\log n}$}{Omega(log n/log log n)} (adaptive) lower bound for \EVAL}\label{sec:tlogn:eval:lb}
While the above lower bound is tight, it only applies to non-adaptive testers; and it is natural to ask whether allowing adaptivity enables one to bypass this impossibility result -- and, maybe, to get constant-query testers. The following result shows that it is not the case: even for constant $\eps$, adaptive testers must also make (almost) logarithmically many queries in the worst case.
\begin{theorem}\label{theo:tester:monotonicity:eval:lb:adaptive}
There exist absolute constants $\eps_0>0$ and $c > 0$ such that the following holds. Any $\eps_0$-tester for monotonicity in the \EVAL model must perform $c\frac{\log n}{\log\log n}$ queries. (Furthermore, one can take $\eps_0=1/2$.)
\end{theorem}
\makeatletter{}\noindent Intuitively, if one attempts to design hard instances for this problem against adaptive algorithms, one has to modify a \yes-instance to get a \no-instance by ``removing'' some probability weight and ``hiding'' it somewhere else (where it will then violates monotonicity). The difficult part in doing so does \emph{not} lie in hiding that extra weight: one can always choose a random element $k\in\{n/2,..,n\}$ and add some probability to it. Arguing that any \EVAL algorithm cannot find $k$ unless it makes $\bigOmega{n}$ queries is then not difficult, as it is essentially tantamount to finding a needle in a haystack.

Thus, the key is to \emph{take} some probability weight from a subset of points of the support, in order to redistribute it. Note that this cannot this time be a local modification, as in a monotone distribution one cannot obtain $\bigOmega{1}$ weight from a constant number of points unless these are amongst the very first elements of the domain; and such case is easy to detect with $\bigO{1}$ queries. Equivalently, we want to describe how to obtain two non-negative monotone sequences that are hard to distinguish, one summing to one (i.e., already being a probability distribution) and the other having sum bounded away from one (the slack giving us some ``weight to redistribute''). To achieve this, we will rely on the following result due to Sariel Har-Peled \cite{HarPeled:CS:Overflow:15}, whose proof is reproduced below:\footnote{For more results on approximating the discrete integral of sequences, including an upper bound for monotone sequences reminiscent of Birg\'e's oblivious decomposition, one may consult \cite{HarPeled:06}.}
\begin{proposition}\label{prop:sariel:sum:monotone:lb}
Given query access to a sequence of non-negative numbers $a_n \geq \dots\geq a_1$ and $\eps \in (0,1)$, along with the promise that either $\sum_{k=1}^n a_k = 1$ or $\sum_{k=1}^n a_k \leq 1-\eps$, any (possibly adaptive) randomized algorithm that distinguishes between the two cases with probability at least $2/3$ must make $\bigOmega{\frac{\log n}{\log\log n}}$ queries in the worst case. (Moreover, the result even holds for $\eps=1/2$).
\end{proposition}
\begin{proof}
Let $(a_k)_{k\in[n]}$ be a sequence defined as follows: we partition the sequence into $L$ blocks. In the $i$-th block there are going to be $n_i$ elements (i.e., $\sum_i n_i = n$).
Set the $i$-th block size to be $n_i = L^i$, where $L \eqdef \bigTheta{ \log n / \log \log n }$ is the number of blocks. Let $\beta\eqdef(2L-1)/({2L})$ be a normalizing factor; an element in the $i$-th block has value $\alpha_i = \frac{\beta}{2Ln_i}$, so that the total sum of the values in the sequence is $\beta/2 < 1/2$.

From $(a_k)_{k\in[n]}$, we obtain another sequence $(b_k)_{k\in[n]}$ by picking uniformly at random an arbitrary block, say the $j$-th one, and set all values in its block to be $\alpha_{j-1} = L \alpha_j$ (instead of $\alpha_j$). This increases the contribution of the $j$-th block from $\beta/2L$ to $\beta/2$, and increase the total sum of the sequence to $\beta(1-\frac{1}{2L}) = 1$. Furthermore, it is straightforward to see that both sequences are indeed non-decreasing.

Informally, the idea is that to distinguish $(a_k)$ from $(b_k)$, any randomized algorithm must check the value in each one of the blocks. As such, it must read at least $\bigOmega{L}$ values of the sequence. To make the above argument more formal, with probability $p=1/2$, give the original sequence of sum $1$ as the input (we refer to this as original input). Otherwise, randomly select the block that has the increased values (modified input). Clearly, if the randomized algorithm reads less than, say, $L/8$ entries, it has probability (roughly) $1/8$ to detect a modified input.  As such, the probability this algorithm fails, if it reads less than $L/8$ entries, is at least $(1-p)(7/8) > 7/16 > 1/3$.
\end{proof}
\begin{proofof}{\autoref{theo:tester:monotonicity:eval:lb:adaptive}}
To get \autoref{theo:tester:monotonicity:eval:lb:adaptive} from \autoref{prop:sariel:sum:monotone:lb}, we define a reduction in the obvious way:  any \EVAL monotonicity tester $\Tester$ can be used to solve the promise problem above by first choosing uniformly at random an element $k$ in $\{2,\dots,n\}$, and then answering any query $j\in[n]\setminus\{k\}$ from $\Algo$ by returning the value $a_j$. (This indeed defines a probability distribution that is either monotone (if $\sum_k a_k = 1$) or far from it  (if $\sum_k a_k = 1/2$): $k$ is the index where the -- possibly -- extra weight $1/2$ would have been ``hidden,'' in a \no-instance; and is therefore the only query point we cannot answer.) Conditioning on $k$ not being queried (which occurs with probability $1-O(1/n)$ given the random choice of $k$, it is straightforward to see that outputting the value returned by \Tester yields the correct answer with probability $2/3$. From the above, any such \Tester must therefore have query complexity $\bigOmega{\frac{\log n}{\log\log n}}$.
\end{proofof}

\paragraph*{Open question.} It is worth noting that a different construction, also due to \cite{HarPeled:CS:Overflow:15}\ignore{ (with this time $\log n$ consecutive blocks of size increasing by a factor 2)}, yields a different lower bound of $\bigOmega{1/\eps}$ for the promise problem of \autoref{prop:sariel:sum:monotone:lb}. Combining the two (and applying the same reduction as above), we obtain a lower bound of $\bigOmega{\max(\log n/\log\log n, 1/\eps)}$ for testing monotonicity in the \EVAL model. However, we do conjecture the right dependence on $n$ to be logarithmic; more specifically, the author believe the above upper bound to be tight:
\begin{conjecture}\label{conjecture:monotonicity:eval}
Monotonicity testing in the \EVAL model has query complexity $\bigOmega{\frac{\log{n}}{\eps}}$.
\end{conjecture}
 
\section{With \texorpdfstring{\Cdfsamp}{cumulative dual }access}\label{sec:polyeps:extended:tester}
\makeatletter{}
\begin{theorem}\label{theo:tester:monotonicity:extension:cdf}
There exists an $\tildeO{\frac{1}{\eps^4}}$-query (independent of $n$) tester for monotonicity in the \Cdfsamp model.
\end{theorem}

\begin{proof}
We first give the overall structure of the tester -- without surprise, very similar to the ones in \autoref{sec:polyeps:cond:tester} and \autoref{sec:polylogneps:eval:ub}:
\begin{algorithm}[H]
  \begin{algorithmic}[1]
    \Require \CDFEVAL and \SAMP access to $\D$
    \State\label{algo:cdf:2:step:0} Set $\alpha\eqdef\eps/4$, and compute $\mathcal{I}_\alpha$.
    \State\label{algo:cdf:2:step:1} Test if $\Dflat{\alpha}$ is $(\eps/4)$-close to monotone by testing $(\eps/4)$-closeness (of $\Dred_{\alpha}$) to $\property_{\alpha}$; \Return \fail if the tester rejects.
    \State\label{algo:cdf:2:step:2} Get an estimate $\hat{d}$ of $\totalvardist{\D}{\Dflat{\alpha}}$ up to additive $\eps/4$; \Return \fail if $\hat{d} > \eps/2$.
    \State\Return\accept
  \end{algorithmic}
  \caption{\label{algo:cdf:monotonicity:algorithm}Algorithm \textsc{TestMonCumulative}}
\end{algorithm}
Before diving into the actual implementation of Steps~\ref{algo:cdf:2:step:1} and \ref{algo:cdf:2:step:2}, we first argue that, conditioned on their outcome being correct, \textsc{TestMonCumulative} outputs the correct answer. The argument is almost identical as in the proof of \autoref{theo:tester:monotonicity:eval}:
\begin{itemize}
  \item if $\D\in\mathcal{M}$, then $\Dflat{\alpha}\in\mathcal{M}$ as well; by \autoref{remark:property:and:expanded:monotonicity} it follows that $\Dred_{\alpha}\in\property_{\alpha}$ and we pass the first step. We also know (\autoref{theorem:birge:obl:decomp}) that $\totalvardist{\D}{\Dflat{\alpha}}\leq \alpha$, so that our estimate satisfies $\hat{d}\leq \alpha+\eps/4 = \eps/2$. Therefore, the algorithm does not reject here either, and eventually outputs \accept.
  \item conversely, if the algorithm outputs \accept, then we have both that \textsf{(a)} the distance of $\Dflat{\alpha}$ to $\mathcal{M}$ is at most $\eps/4$, and \textsf{(b)} $\totalvardist{\D}{\Dflat{\alpha}}\leq \hat{d} + \eps/4 \leq 3\eps/4 $; so overall $\totalvardist{\D}{\mathcal{M}}\leq \eps$.
\end{itemize}
It remains to show how to perform steps \ref{algo:cdf:2:step:1} and \ref{algo:cdf:2:step:2} -- namely, testing $\Dred_{\alpha}$ for $\property_{\alpha}$ given \CDFEVAL and \SAMP access to $\D$, and approximating $\totalvardist{\D}{\Dflat{\alpha}}$.

\paragraph{Testing \texorpdfstring{$\gamma$-closeness to $\property_\alpha$}{closeness to P}}

This part is performed similarly as in \autoref{ssec:testing:dist:to:exp:property}, observing that one can easily simulate access to $\PCOND_Q$  from a $\CDFEVAL_Q$  oracle.  Indeed, \autoref{lemma:distance:palpha::weight:witnesses} implies that when $Q$ is \eps-far from having the property, it suffices to sample $\bigO{\frac{1}{\alpha\eps}}$ points according to $Q=\Dred_{\alpha}$ and compare them to their neighbors to detect a violation with probability at least $9/10$. Note that this last test is easy, as we have query access to $Q$ (recall that we have a $\CDFEVAL_\D$ oracle, and that $\Dred_{\alpha}(k)=\D(I_k)$).

\paragraph{Efficient approximation of distance to \texorpdfstring{$\Dflat{}$}{Dflat}}
Let $\D$, \eps and $\mathcal{I}_{\alpha}$ be as before; define $Z$ to be a random variable taking values in $[0,1]$, such that, for $k\in[\ell]$, $Z$ is equal to $\totalvardist{D_{I_k}}{\uniform_{I_k}}$ with probability $w_k=D(I_k)$. It follows that
\begin{align}
  \shortexpect Z &= \sum_{k=1}^\ell w_k \totalvardist{D_{I_k}}{\uniform_{I_k}} = \half\sum_{k=1}^\ell w_k \sum_{i\in I_k}^\ell \abs{D_{I_k}(i)-\frac{1}{\abs{I_k}}} \notag \\
  &= \half\sum_{k=1}^\ell \sum_{i\in I_k} \abs{D(i)-\frac{D(I_k)}{\abs{I_k}}} = \half\sum_{i=1}^n \abs{D(i)-\Dflat{\alpha}} \notag \\
  &= \totalvardist{\D}{\Dflat{\alpha}}.
\end{align}
Furthermore, one can simulate $m=\bigO{1/\eps^2}$ i.i.d. draws from $Z$ by repeating independently the following for each of them:
\begin{itemize}
  \item draw $i\sim D$ by calling $\SAMP_D$, and look up the $k$ for which $i\in I_k$;
  \item get the value $\D(I_k)$ with 2 $\CDFEVAL$ queries (note that $\D(I_k)> 0$, as we just got a sample from $I_k$);
  \item estimate $\totalvardist{D_{I_k}}{\uniform_{I_k}}$ up to $\pm\eps$ (with failure probability at most $\frac{1}{10m}$)   by drawing $\tildeO{1/\eps^2}$ uniform samples from $I_k$ and querying the values of $\D$ on them, to estimate
  \begin{align*}
    \totalvardist{D_{I_k}}{\uniform_{I_k}} &= \sum_{j\in I_k} \abs{D_{I_k}(j)-\frac{1}{\abs{I_k}}}\indic{\abs{I_k}D_{I_k}(j)< 1} = \sum_{i\in{I_k}} \abs{\frac{D(j)}{D(I_k)}\abs{I_k}-1}\cdot \frac{\indic{\abs{I_k}D_{I_k}(j)< 1}}{\abs{I_k}}  \\
    &= \shortexpect_{j\sim \uniform_{I_k}}\! \left[ \abs{  \frac{D(j)}{D(I_k)}\abs{I_k} -1}\cdot\indic{\frac{D(j)}{D(I_k)}\abs{I_k}< 1} \right]
  \end{align*}
\end{itemize}
Applying a Chernoff bound (and a union bound over all such simulated draws), performing this only $m$ times is sufficient to get an \eps-additive estimate of $\shortexpect Z$ with probability at least $9/10$, for an overall $\tildeO{1/\eps^{4}}$ number of queries (sample and evaluation).

\paragraph{Proof of \autoref{theo:tester:monotonicity:extension:cdf}}
Correctness has already been argued, provided both subroutines do not err; by a union bound, the ``good event'' that the two steps produce such an outcome happens with probability at least $4/5$.
\noindent The query complexity is the sum of $\bigO{\frac{1}{\alpha\eps}}=\bigO{\frac{1}{\eps^2}}$ (Step \ref{algo:cdf:2:step:1}) and $\tildeO{1/\eps^{4}}$ (Step \ref{algo:cdf:2:step:2}), yielding overall the claimed $\tildeO{1/\eps^{4}}$ sample complexity.
\end{proof}

\paragraph{Acknowledgments}

The bulk of \autoref{sec:polyeps:cond:tester} originates from discussions and correspondence with Dana Ron and Rocco Servedio, and owes its existence to them; similarly, the contents of \autoref{sec:polyeps:extended:tester} derive from work with Ronitt Rubinfeld. However, errors, mistakes and typos are of the author's own device, and he claims sole ownership of these.
\clearpage
 \bibliographystyle{alpha}
 \bibliography{monotonedistributions} 
 
 \clearpage
 \appendix\addtocontents{toc}{\protect\setcounter{tocdepth}{1}}
 \makeatletter{}\section{Models of distribution testing: definitions}\label{app:def:models}
In this appendix, we define precisely the notion of \emph{testing} of distributions over a domain $\domain$ (e.g., $\domain=[n]$), and the different models of access covered in this work.\medskip

Recall that a \emph{\indexed{property}} $\property$ of distributions over $\domain$ is a set consisting of all distributions that have the property. The distance from $\D$ to a property $\property$, denoted $\totalvardist{\D}{\property}$, is then defined as $\inf_{\D^\prime \in \property} \totalvardist{\D}{\property}$.
We use the standard definition of testing algorithms for properties of distributions over $\domain$, where $n$ is the relevant parameter for $\domain$ (i.e., in our case, its size $\abs{\domain}$). We chose here to phrase it in the most general setting possible with regard to how the unknown distribution is ``queried'': and will specify this aspect further in the next paragraphs (sampling access, conditional access, etc.).
\begin{definition}\label{def:testing}\index{testing algorithm}
  Let $\property$ be a property of distributions over $\domain$.  Let $\ORACLE_\D$ be an oracle providing some type of access to $\D$. A \emph{$q$-query $\ORACLE_\D$ testing algorithm for $\property$} is a randomized algorithm $\Tester$ which takes as input $n$, $\eps\in(0,1]$, as well as access to $\ORACLE_\D$.  After making at most $q(\eps,n)$ calls to the oracle, \Tester either outputs \accept or \reject, such that the following holds:
  \begin{itemize}
  \item if $\D \in \property$, \Tester outputs \accept with probability at least $2/3$;
  \item if $\totalvardist{\D}{\property} \geq \eps$, \Tester outputs \reject with probability at least $2/3$.
  \end{itemize}
\end{definition}

\subsection{The \SAMP model}
In this first and most common setting, the testers access the unknown distribution by getting independent and identically distributed (\iid) samples from it.
\begin{definition}[Standard access model (sampling)]\label{def:sampling:oracle}
    Let $\D$ be a fixed distribution over $\domain$. A \emph{sampling oracle for $\D$} is an oracle $\SAMP_\D$ defined as follows: when queried, $\SAMP_\D$ returns an element $x\in\domain$, where the probability that $x$ is returned is $\D(x)$ independently of all previous calls to the oracle.
\end{definition}
\noindent This definition immediately implies that all algorithms in this model are by essence non-adaptive: indeed, any tester or tolerant tester can be converted into a non-adaptive one, without affecting the sample complexity. (This is a direct consequence of the fact that all an adaptive algorithm can do when interacting with a \SAMP oracle is deciding to stop asking for samples, based on the ones it already got, or continue.)

\subsection{The \COND model}
\begin{definition}[Conditional access model \cite{CFGM:13,CRS:12}]\label{def:conditional:oracle}
Fix a distribution $\D$ over $\domain$.  A \emph{\COND oracle for $\D$}, denoted
$\COND_\D$, is defined as follows:
 the oracle takes as input a \emph{query set}
 $S \subseteq \domain$, chosen by the algorithm,  that has $\D(S) > 0$. The oracle returns an element $i \in S$, where
 the probability that element $i$ is returned is $\D_S(i) = \D(i)/\D(S),$
 independently of all previous calls to the oracle.
\end{definition}
Note that as described above the behavior of $\COND_\D(S)$ is undefined if $\D(S)=0$, i.e., the set $S$ has zero probability under $\D$.  Various definitional choices could be made to deal with this: e.g., Canonne et al. \cite{CRS:12} assume that in such a case the oracle (and hence the algorithm) outputs ``failure'' and terminates, while Chakraborty et al. \cite{CFGM:13} define the oracle to return in this case a sample uniformly distributed in $S$. In most situations, this distinction does not make any difference, as most algorithms can always include in their next queries a sample previously obtained; however, the former choice does rule out the possibility of \emph{non-adaptive} testers taking advantage of the additional power \COND provides over \SAMP; while such testers are part of the focus of \cite{CFGM:13}.\medskip

Testing algorithms can often only be assumed to have the ability to query sets $S$ that have some sort of ``structure,'' or are in some way ``simple.'' To capture this, one can define specific restrictions of the general $\COND$ model, which do not allow \emph{arbitrary} sets to be queried but instead enforce some constraints on the queries: \cite{CRS:12} introduces and studies two such restrictions, ``\PCOND''  and ``\ICOND.''
\begin{definition}[Restricted conditional oracles]\label{def:conditional:oracle:pcond:icond}
A \emph{\PCOND} (``pair-cond'')
\emph{oracle for $\D$} is a restricted version
of $\COND_\D$ that only accepts input sets $S$ which are either
$S=\domain$ (thus providing the power of a $\SAMP_\D$ oracle)
or $S=\{x,y\}$ for some $x,y \in \domain$, i.e. sets of size two.

In the specific case of $\domain=[n]$, an \emph{\ICOND} (``interval-cond'')
\emph{oracle for $\D$} is a restricted version
of $\COND_\D$ that only accepts input sets $S$ which are intervals
$S=[a,b]=\{a,a+1,\dots,b\}$ for some $a \leq b \in [n]$ (note that
taking $a=1,$ $b=n$ this provides the power of a $\SAMP_{\D}$ oracle).
\end{definition}

\subsection{The \EVAL, \Pdfsamp and \Cdfsamp models}

    \begin{definition}[Evaluation model \cite{RS:09}]
    Let $\D$ be a fixed distribution over $\domain$. An \emph{evaluation oracle for $\D$} is an oracle $\EVAL_\D$ defined as follows: the oracle takes as input a query element $x\in\domain$, and returns the probability weight $\D(x)$ that the distribution puts on $x$.
    \end{definition}
    \begin{definition}[\Pdfsamp access model \cite{BDKR:05,GMV:06,CR:14}]
    Let $\D$ be a fixed distribution over $\domain$. A \emph{\pdfsamp oracle for $\D$} is a pair of oracles $(\SAMP_\D, \EVAL_\D)$ defined as follows: when queried, the \emph{sampling} oracle $\SAMP_\D$ returns an element $x\in\domain$, where the probability that $x$ is returned is $\D(x)$ independently of all previous calls to any oracle; while the \emph{evaluation} oracle $\EVAL_\D$ takes as input a query element $y\in\domain$, and returns the probability weight $\D(y)$ that the distribution puts on $y$.
    \end{definition}
    \begin{definition}[\Cdfsamp access model \cite{BKR:04,CR:14}]
    Let $\D$ be a fixed distribution over $\domain=[n]$. A \emph{\cdfsamp oracle for $\D$} is a pair of oracles $(\SAMP_\D, \CDFEVAL_\D)$ defined as follows: the \emph{sampling} oracle $\SAMP_\D$ behaves as before, while the \emph{evaluation} oracle $\CDFEVAL_\D$ takes as input a query element $j\in[n]$, and returns the probability weight that the distribution puts on $[j]$, that is $\D([j])=\sum_{i=1}^j D(i)$ .
    \end{definition}  
    (Note that, as the latter requires some total ordering on the domain, it is only defined for distributions over $[n]$; as was the \ICOND oracle from \autoref{def:conditional:oracle:pcond:icond}.)
 
 \makeatletter{}\section{Tolerant testing with \Pdfsamp and \Cdfsamp access}\label{app:tolerant:extended}
In this appendix, we show how similar ideas can yield \emph{tolerant} testers for monotonicity, as long as the access model admits both an agnostic learner for monotone distributions and an efficient distance approximator. As this is the case for both the \Pdfsamp and \Cdfsamp oracles, this allows us to derive such tolerant testers with logarithmic query complexity. (Note that these results imply that, in both models, tolerant testing monotonicity of an arbitrary distribution is no harder than learning an \emph{actually monotone} distribution.)\footnote{It is however worth noting that, because of the restriction $\eps_2 > 3\eps_1$ it requires, our tolerant testing result does \emph{not} imply a distance estimation algorithm.}
\begin{theorem}\label{theo:tolerant:tester:monotonicity}
Fix any constant $\gamma > 0$. In the \pdfsamp access model, there exists an $(\eps_1, \eps_2)$-tolerant tester for monotonicity with query complexity $\bigO{\frac{\log n}{\eps_2^3}}$, provided $\eps_2 > (3+\gamma)\eps_1$.
\end{theorem}
\begin{proof}
We will use here the result of \autoref{lemma:estimate:tolerant:identity} (with its probability of success slightly increased to 5/6 by standard techniques), as well as the Birg\'e decomposition of $[n]$ (with parameter $\bigOmega{\eps_2}$) from \autoref{def:birge:obl:decomp}. The tester is described in \autoref{algo:tolerant:monotonicity:algorithm}, and follows a very simple idea: leveraging the robustness of the Birg\'e flattening, it first agnostically learns an approximation $\bar{\D}$ of the distribution and computes (offline) its distance to monotonicity. Then, using the (efficient) tolerant identity tester available in both \pdfsamp and \cdfsamp models, it estimates the distance between $\bar{\D}$ and $\D$: if both distances are small, the triangle inequality allows us to conclude $\D$ must be close to monotone.

\begin{algorithm}[h!]
  \begin{algorithmic}[1]
    \Require $\SAMP_{\D}$ and $\EVAL_{\D}$ oracle access, parameters $0\leq \eps_1 < \eps_2$ such that $\eps_2 > (3+\gamma)\eps_1$
    \State Set $\alpha\eqdef\frac{\gamma}{6(3+\gamma)}$, $\gamma_1\eqdef 2\eps_1+2\alpha\eps_2$ and $\gamma_2\eqdef \mleft(1-\alpha\mright)\eps_2-\eps_1$. \Comment{So that $\gamma_2 - \gamma_1 = \bigOmega{\eps_2}$}
    \State\label{algo:monot:step:1} Learn $\Dflat{\alpha\eps_2}$ to distance $\alpha\eps_2$, getting a (piecewise constant) $\bar{\D}$.          \Comment{$\bigO{\tfrac{\log n}{\eps_2^3}}$ samples}
    \State\label{algo:monot:step:2} Test if $\bar{\D}$ is $(\eps_1+\alpha\eps_2)$-close to monotone; \Return \fail if not.   \Comment{no sample needed (LP)}
    \State\label{algo:monot:step:3} Test if $\bar{\D}$ is $\gamma_1$-close to $D$ vs. $\gamma_2$-far; \Return \fail if far. \Comment{Tester from \autoref{lemma:estimate:tolerant:identity}.}
    \State\Return\accept
  \end{algorithmic}
  \caption{\label{algo:tolerant:monotonicity:algorithm}Algorithm \textsc{TolerantTestMonotonicity}}
\end{algorithm}

  \paragraph*{Correctness.} 
  Suppose the execution of each step is correct; by standard Chernoff bounds, this event for Step~\ref{algo:monot:step:1} (that is, our estimate $\bar{\D}$ of $\Dflat{\eps_1}$ being indeed $\eps_1$-accurate)  happens with probability at least 5/6; furthermore, Step~\ref{algo:monot:step:2} is deterministic, and Step~\ref{algo:monot:step:3} only fails with probability 1/6 -- so the overall ``good event'' happens with probability at least 2/3. We hereafter condition on this.
,    \begin{itemize}
      \item if $\totalvardist{\D}{\mathcal{M}}\leq \eps_1$, then according to \autoref{coro:birge:decomposition:robust}:
        \begin{enumerate}[(a)]
          \item $\totalvardist{\Dflat{\alpha\eps_2}}{\mathcal{M}}\leq \eps_1$, so that 
   $\totalvardist{\bar{\D}}{\mathcal{M}}\leq \eps_1+ \alpha\eps_2$ and we pass Step~\ref{algo:monot:step:2}; and
          \item $\totalvardist{\D}{\Dflat{\alpha\eps_2}}\leq 2\eps_1+\alpha\eps_2$, which in turn implies that
           $\totalvardist{\D}{\bar{\D}}\leq 2\eps_1+2\alpha\eps_2=\gamma_1$; and we pass Step~\ref{algo:monot:step:3}.
        \end{enumerate}
        Therefore, the tester eventually outputs \accept.
      \item Conversely, if the tester accepts, it means that
        \begin{enumerate}[(a)]
          \item $\totalvardist{\bar{\D}}{\mathcal{M}}\leq \eps_1 + \alpha\eps_2$ and
          \item $\totalvardist{\D}{\bar{\D}}\leq \gamma_2 = \mleft(1-\alpha\mright)\eps_2-\eps_1$
        \end{enumerate}
        hence, by the triangle inequality $\totalvardist{\D}{\mathcal{M}} \leq \eps_2$.
    \end{itemize}
    
    \paragraph*{Query complexity.} 
    The algorithm makes $\bigO{\frac{\log n}{\eps_2^3}}$ \SAMP queries in Step~\ref{algo:monot:step:1}, and $m=\bigO{\frac{1}{(\gamma_2-\gamma_1)^2}}=\bigO{\frac{1}{\eps_2^2}}$ \EVAL queries in Step~\ref{algo:monot:step:3}.
\end{proof}
\begin{remark}
  With $\CDFEVAL$ access (instead of $\EVAL$), the cost of Step~\ref{algo:monot:step:1} would be reduced from $\bigO{{\log n}/{\eps_2^3}}$ sample queries to get an $(\alpha\eps_2)$-approximation to $\bigO{{\log n}/{\eps_2}}$ cdf queries to get the \emph{exact} flattened distribution (as we have only that many quantities of the form $\D(I_k)$ to learn, and each of them requires only 2 cdf queries). This leads to the following corollary:
\end{remark}
\begin{corollary}\label{theo:tolerant:tester:monotonicity:cdf}
Fix any constant $\gamma > 0$. In the \Cdfsamp access model, there exists an $(\eps_1, \eps_2)$-tolerant tester for monotonicity with query complexity $\bigO{\frac{1}{\eps_2^2} + \frac{\log n}{\eps_2}}$, provided $\eps_2 > (3+\gamma)\eps_1$.
\end{corollary}

\end{document}